\documentclass[11pt,letterpaper]{article}
\usepackage[margin=1in]{geometry}

\usepackage{url}
\usepackage{color}
\usepackage[dvipsnames]{xcolor}
\usepackage{array}
\usepackage{algorithmicx}
\usepackage{xy}
\usepackage{multicol}
\usepackage{hyperref}
\usepackage{comment}
\usepackage{enumitem}
\usepackage{float}
\usepackage{algorithm}

\usepackage[noend]{algpseudocode}

\usepackage{tikz}
\usepackage{pgfplots}
\pgfplotsset{compat=1.18} 
\usepackage{subcaption}
\usepgfplotslibrary{groupplots}
\usepgfplotslibrary{fillbetween}
\usepackage{pgfplotstable}

\usetikzlibrary{positioning}
\usepackage{amsthm}
\usepackage{amsmath}
\usepackage{amssymb}
\usepackage{cleveref}
\usepackage[subtle,mathspacing=normal]{savetrees}

\newtheorem{theorem}{Theorem}
\newtheorem{definition}[theorem]{Definition}
\newtheorem{lemma}[theorem]{Lemma}

\newtheorem{corollary}[theorem]{Corollary}

\theoremstyle{remark}

\renewcommand{\paragraph}[1]{\vspace{.2 cm} \noindent \textbf{#1}}

\setlist[itemize]{noitemsep}
\setlist[itemize]{nolistsep}
\setlist[enumerate]{noitemsep}
\setlist[enumerate]{nolistsep}

\usepackage{mathtools} 
\usepackage{xspace}

\newcommand{\defn}[1]{\textbf{\emph{#1}}}

\renewcommand{\epsilon}{\varepsilon}

\algnewcommand\algorithmicforeach{\textbf{for each}}
\algdef{S}[FOR]{ForEach}[1]{\algorithmicforeach\ #1\ \algorithmicdo}
\algnewcommand\algorithmicnot{\textbf{not}}
\algdef{S}[IF]{IfNot}[1]{\algorithmicif\ \algorithmicnot\ #1\ \algorithmicthen}
\makeatletter
\renewcommand{\ALG@name}{Subroutine}
\makeatother

\begin{document}

\title{Bounding the Fragmentation of B-trees Subject to Batched Insertions}

\author{Michael A. Bender \\
\small Stony Brook University and RelationalAI \\ 
 \small\texttt{bender@cs.stonybrook.edu}
\and
Aaron Bernstein \\
\small New York University \\
\small\texttt{aaron.bernstein@nyu.edu}
\and
Nairen Cao\\
\small New York University \\
\small\texttt{nc1827@nyu.edu}
\and
Alex Conway\\
\small Cornell Tech \\
\small\texttt{me@ajhconway.com}
\and
Mart\'in Farach-Colton \\
\small New York University \\
\small\texttt{martin@farach-colton.com}
\and
Hanna Koml\'os \\
\small Max Planck Institute for Informatics \\
\small\texttt{hkomlos@gmail.com}
\and
Yarin Shechter \\
\small New York University \\
\small\texttt{yahrin4@gmail.com}
\and
Nicole Wein \\
\small University of Michigan\\
\small\texttt{nswein@umich.edu}
}

\date{}

\maketitle

\begin{abstract}
The issue of internal fragmentation in data structures is a fundamental challenge in database design. A seminal result of Yao in this field shows that evenly splitting the leaves of a B-tree against a workload of uniformly random insertions achieves space utilization of around 69\%. However, many database applications perform batched insertions, where a small run of consecutive keys is inserted at a single position. We develop a generalization of Yao’s analysis to provide rigorous treatment of such batched workloads. Our approach revisits and reformulates the analytical structure underlying Yao’s result in a way that enables generalization and is used to argue that even splitting works well for many workloads in our extended class. For the remaining workloads, we develop simple alternative strategies that provably maintain good space utilization.
\end{abstract}

\section{Introduction}
\label{sec:intro}
The issue of internal fragmentation in data structures is a fundamental challenge in database design~\cite{DBLP:journals/tos/KesavanCDM20,DBLP:journals/tsmc/TamhankarR98,fritchey2012fragmentation,andriyani2024maintaining}.  The issue is that when a leaf of a B-tree fills, it must be split (evenly or otherwise), and the average fill of the resulting leaves is 50\%, no matter how the leaves are split. Underfull  leaves cause a number of well-studied performance degradations, including more cache misses~\cite{DBLP:conf/damon/KuhnBHCT23}, space management on NVMe~\cite{DBLP:conf/cidr/AlhomssiL21}, computational slowdown in in-memory databases~\cite{DBLP:journals/is/KwonLNNPCM23}, tree rebalancing costs~\cite{DBLP:journals/corr/abs-2505-01180}, etc.

Studying the fill of the leaves is called \emph{fringe analysis}~\cite{DBLP:journals/iandc/WoodEZM82,DBLP:journals/csur/Baeza-Yates95}, and for B-trees, the fill of the leaves is the high-order term in the overall fill because all but a $\Theta(1/B)$ fraction of nodes are leaves. 
In this paper, we study the fill of B-tree leaves,
specifically, we define the \defn{block-splitting problem} as follows. Insertions arrive online, where each insertion is at a specified \defn{rank} (between two previously inserted elements, or at the beginning or end). The algorithm must maintain a partition of the elements into \defn{blocks} of size at most $B$ each, where each block is a contiguous set of elements. That is, once a block reaches size $B+1$, the algorithm must irrevocably split it into two blocks. The \defn{fill} achieved by an algorithm (after sufficiently many insertions) is the average over all blocks of the number of elements in the block divided by $B$.

A critical special case for the block-splitting problem  is when insertions arrive one at a time at uniformly random ranks.  Splitting blocks evenly results in a distributions of fills from 50\% full to completely full.  If this distribution were even, the average fill of a leaf would be 75\%, but fuller blocks are more likely to receive insertions than emptier blocks.  Thus, the true answer should be less than 75\% but greater than 50\%.   

A seminal result of Yao from 1978~\cite{DBLP:journals/acta/Yao78,DBLP:journals/corr/abs-2202-04185} shows that evenly splitting the leaves of a B-tree under uniform insertions  results in an average fill of $\ln 2 \approx 69\%$.    This model of insertions is quite common in real databases, as has been reported by Facebook~\cite{DBLP:conf/cidr/DongCGBSS17}, InnoDB~\cite{dubois2013mysql}, Oracle~\cite{devCheckFragmentation}, and Microsoft SQL Server~\cite{sqlservercentralUnderstandingCRUD}.  

For many databases, however, uniformly random insertions do not accurately model typical usage; instead, databases often witness \defn{batches} of sequential insertions~\cite{DBLP:journals/corr/abs-2003-01064,DBLP:journals/pvldb/ChaHWZAY23,DBLP:journals/isci/Kim02,DBLP:journals/pvldb/AchakeevS13,DBLP:conf/sigmod/WangPLLZKA18}.  In this case, an insertion operation involves inserting $r$ items,  one after another, sequentially in key space.  
Many databases and key-value stores have specific heuristics for increasing the fill of leaves in the case of batch insertions~\cite{perconaPerconaServer,microsoftBehindScenes,DBLP:journals/ftdb/Graefe24}.

In this paper, we generalize Yao's approach to the batch insertion case: with a fixed batch size $r$, a uniformly random rank is (repeatedly) selected and $r$ sequential insertions are performed at that rank.   It might seem that as $r$ grows, the problem of achieving high fill would get easier.  After all, we known more information about the insertion pattern into the database as more items are inserted in each location. However, we find that as least for even splitting and the closely related \defn{deferred even splitting}, which we define below, the picture isn't so clear.  Before we explore the complexity of the varying $r$, we describe a few splitting algorithms:

\begin{itemize}
   \item \textbf{Even splitting.}
    We process the \(r\) new elements one by one. Whenever a block grows from size \(B\) to size \(B+1\), we split it into two blocks of sizes  \(\big\lfloor (B+1)/2 \big\rfloor\) and \(\big\lceil (B+1)/2 \big\rceil\). 
    We will then keep inserting to the block with size \(\big\lceil (B+1)/2 \big\rceil\).

    \item \textbf{Deferred even splitting.}
    Suppose $r$ elements are inserted into a block that starts out with with $\ell$ elements, and that $r+\ell>B$.
    The deferred even split strategy replaces this block with the minimum number of blocks possible so that all resulting blocks differ in fill by at most 1.  That is, we break the block into $\lceil (r+\ell)/B \rceil$ blocks with the same fill, up to rounding.\footnote{The splitting in this strategy is not literally deferred. It can be implemented so that the splitting happens as items get inserted, rather than requiring a block to get overfull before splitting.}

\end{itemize}

\pgfplotsset{
  randomhammer/.style={
    small,
    width=\linewidth,
    height=0.7\linewidth,
    grid=both,
    xlabel={insertion length $r$},
    ylabel={mean fullness},
    legend cell align={left},
    legend pos=south east,
    ymin = 0.5,
    ymax = 1,
    xmin = 0,
    y tick label style={rotate=90},
    ytick distance = 0.1,
  }
}

\pgfmathsetmacro{\B}{240}

\newcommand{\Hdiff}[1]{%
  \ifcase#1
    0
  \or 0.5000000000
  \or 0.5833333333
  \or 0.6166666667
  \or 0.6345238095
  \or 0.6456349206
  \or 0.6532106782
  \or 0.6587051837
  \or 0.6628718504
  \or 0.6661398242
  \or 0.6687714032
  \or 0.6709359053
  \or 0.6727474995
  \or 0.6742859611
  \or 0.6756087124
  \or 0.6767581377
  \or 0.6777662022
  \or 0.6786574678
  \or 0.6794511186
  \or 0.6801623562
  \or 0.6808033818
  \else 0.6808033818
  \fi
}

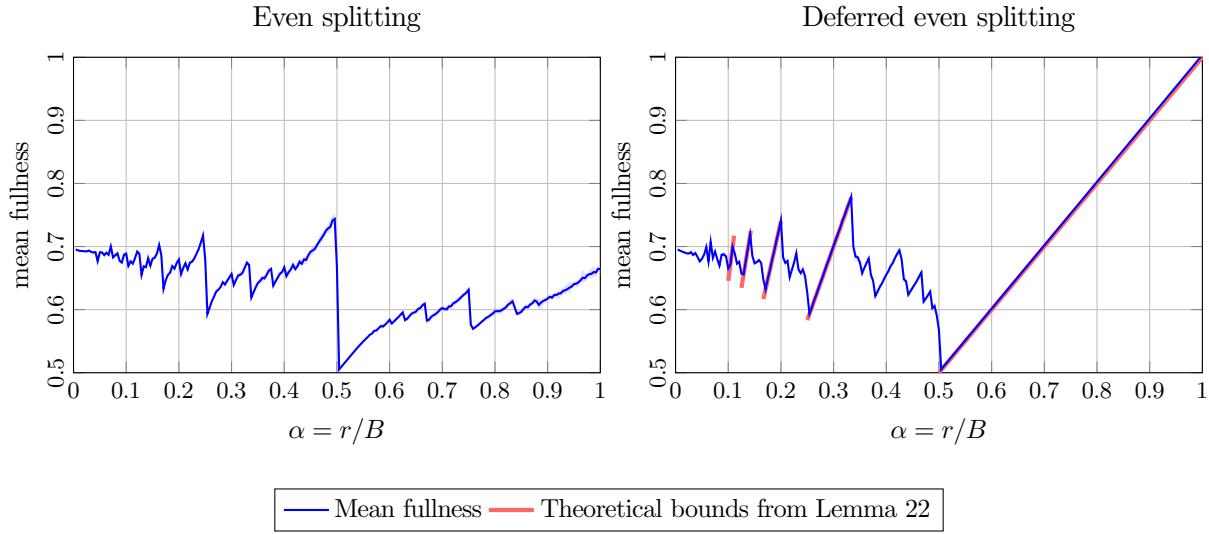
\begin{figure}[t!]
\centering

\begin{tikzpicture}
\begin{groupplot}[
  group style={
    group size=2 by 1,
    horizontal sep=1cm,
    vertical sep=1cm,
  },
  randomhammer,
  width=0.52\textwidth,
  height=0.35\textwidth,
  xmin=0, xmax=1,                
  xtick distance=0.1,
  xlabel={$\alpha = r/B$},
]

\nextgroupplot[
  title={Even splitting},
  legend to name=sharedlegend,
  legend style={nodes={scale=0.9, transform shape}},
  legend columns=2,
]

\addplot[ thick, blue ] coordinates {(0,0)}; 
\addlegendentry{Mean fullness}

\addplot[ ultra thick, red!60 ] coordinates {(0,0)};
\addlegendentry{Theoretical bounds from Lemma~\ref{lem:DeferredEvenSplitting}}

\addplot[name path=upper, draw=none]
  table[x expr=\thisrow{hammer_h}/240, y=max_fullness, col sep=comma]
        {figures/immed-even-240-small-r.csv};

\addplot[name path=lower, draw=none]
  table[x expr=\thisrow{hammer_h}/240, y=min_fullness, col sep=comma]
        {figures/immed-even-240-small-r.csv};

\addplot[fill=blue!30, fill opacity=0.4] fill between[of=upper and lower];

\addplot[thick, blue]
  table[x expr=\thisrow{hammer_h}/240, y=mean_fullness, col sep=comma]
        {figures/immed-even-240-small-r.csv};

\nextgroupplot[
  title={Deferred even splitting}
]

\foreach \i in {1,...,5} {%
  \addplot[
    ultra thick,
    red!60,
    domain={1/(2*\i)}:{1/(2*\i - 1)},   
    samples=2,
    forget plot,
  ]
  { (2*\i/240) * \Hdiff{\i} * x * 240 }; 
}

\addplot[name path=upper, draw=none]
  table[x expr=\thisrow{hammer_h}/240, y=max_fullness, col sep=comma]
        {figures/defer-even-240-small-r.csv};

\addplot[name path=lower, draw=none]
  table[x expr=\thisrow{hammer_h}/240, y=min_fullness, col sep=comma]
        {figures/defer-even-240-small-r.csv};

\addplot[fill=blue!30, fill opacity=0.4] fill between[of=upper and lower];

\addplot[thick, blue]
  table[x expr=\thisrow{hammer_h}/240, y=mean_fullness, col sep=comma]
        {figures/defer-even-240-small-r.csv};

\end{groupplot}
\end{tikzpicture}

\begin{center}
\ref{sharedlegend}
\end{center}

\caption{Experimental fullness of deferred and even splitting on batch insertions on blocks of size $B=240$. 200,000 insertions are made from empty using batch insertions of varying length ($r \in [1,B]$). The mean fullness across 10 independent runs is shown. Bounds from \Cref{lem:DeferredEvenSplitting} are shown in red for comparison.}
\label{fig:intro-exp-small}
\end{figure}

\vspace{1em}
Before stating our results, we demonstrate the surprisingly erratic behavior of these simple algorithms for various values of the batch size $r$ with data from simulations.
Figure~\ref{fig:intro-exp-small}
shows the measured fill achieved by even splitting and deferred even splitting when $r < B$.  In the latter case, we also prove tight bounds on the fill in  Lemma~\ref{lem:DeferredEvenSplitting}, and we see that these bounds align with our experiments. 

We observe two interesting phenomena in Figure ~\ref{fig:intro-exp-small}. 
The first is that the average fill varies quite unpredictably with $r$. Nonetheless,  even splitting  performs well in our simulations, as long as $r$ is relatively small. 

The second observation is that,  unfortunately, according to the simulations, these algorithms do not always get better than the trivial 50\% fill. In particular, for $r=B/2$, both algorithms get exactly 50\% fill. However, we show that this undesirable behavior is not inherent to batch insertions, and interestingly, we can achieve better fill by employing \emph{uneven} splitting.

\begin{itemize}
    \item \textbf{(Deferred) Uneven splitting.}  
        Suppose we are inserting $r$ elements into a block that starts out with $\ell$ elements, and that $r+\ell \in (B,2B)$.
     We split the block into two blocks of sizes     \(\big\lfloor \delta (r+\ell)/2 \big\rfloor\) and \( r+\ell - \big\lfloor \delta (\ell+r)/2 \big\rfloor\), where $\delta \in (0,1)$ is the \defn{splitting factor} for this split.\footnote{In general, $\delta$ can be different at each split.} Note that this rule is only defined when both resulting blocks have at most $B$ items.   
\end{itemize}

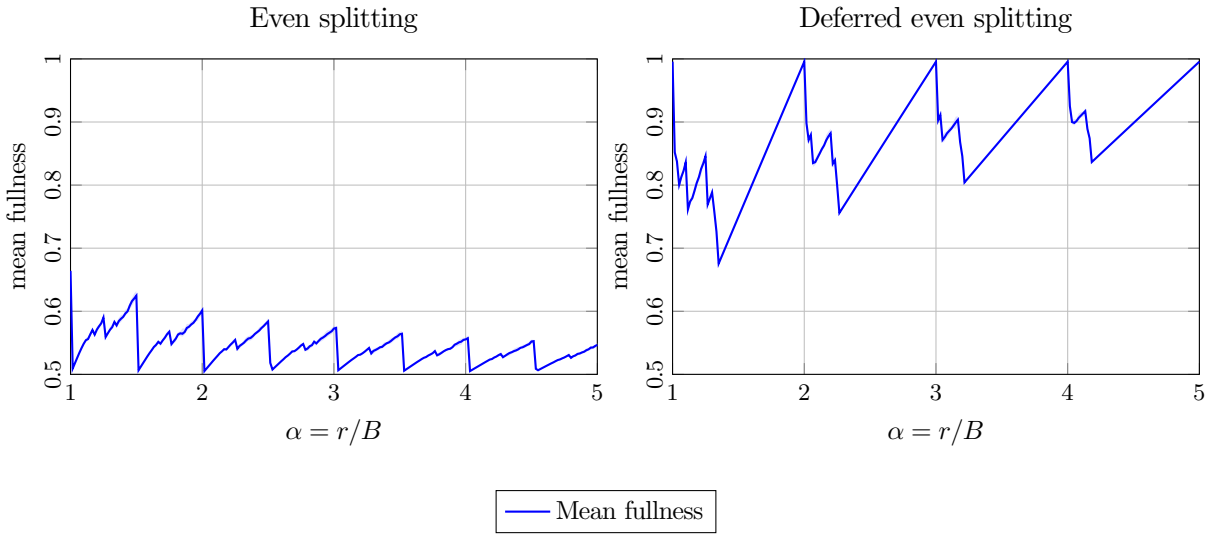
\begin{figure}[h]
\centering

\begin{tikzpicture}
\begin{groupplot}[
  group style={
    group size=2 by 1,
    horizontal sep=1cm,
    vertical sep=1cm,
  },
  randomhammer,
  width=0.52\textwidth,
  height=0.35\textwidth,
  xmin=1,
  xmax=5,
  xtick={1,2,3,4,5},
  xlabel={$\alpha = r/B$},
]

\nextgroupplot[
  title={Even splitting},
  legend to name=sharedlegend,
  legend style={nodes={scale=0.9, transform shape}},
  legend columns=2,
]

\addplot[
  thick,
  blue,
] coordinates {(0,0)};
\addlegendentry{Mean fullness}

\addplot[
  ultra thick,
  red!60,
] coordinates {(0,0)};

\addplot[name path=upper, draw=none]
  table[x expr=\thisrow{hammer_h}/240, y=max_fullness, col sep=comma]
        {figures/immed-even-240-large-r.csv};
\addplot[name path=lower, draw=none]
  table[x expr=\thisrow{hammer_h}/240, y=min_fullness, col sep=comma]
        {figures/immed-even-240-large-r.csv};
\addplot[fill=blue!30, fill opacity=0.4] fill between[of=upper and lower];

\addplot[thick, blue]
  table[x expr=\thisrow{hammer_h}/240, y=mean_fullness, col sep=comma]
        {figures/immed-even-240-large-r.csv};

\nextgroupplot[title={Deferred even splitting}]


\addplot[name path=upper, draw=none]
  table[x expr=\thisrow{hammer_h}/240, y=max_fullness, col sep=comma]
        {figures/defer-even-240-large-r.csv};
\addplot[name path=lower, draw=none]
  table[x expr=\thisrow{hammer_h}/240, y=min_fullness, col sep=comma]
        {figures/defer-even-240-large-r.csv};
\addplot[fill=blue!30, fill opacity=0.4] fill between[of=upper and lower];

\addplot[thick, blue]
  table[x expr=\thisrow{hammer_h}/240, y=mean_fullness, col sep=comma]
        {figures/defer-even-240-large-r.csv};

\end{groupplot}
\end{tikzpicture}

\begin{center}
\ref{sharedlegend}
\end{center}

\caption{Experimental fullness of deferred and even splitting on batch insertions on blocks of size $B=240$. 200,000 insertions are made from empty using batch insertions of varying length ($r \in [B, 5B]$). The mean fullness across 10 independent runs is shown, as a function of $\alpha = r/B$.}
\label{fig:intro-exp-large}
\end{figure}

\phantom{x}

 In Figure~\ref{fig:intro-exp-large}, we consider large values of $r$.  For even splitting, the fullness regularly hits 50\% and trends downward on average as $r$ grows.  For deferred even splitting, the fullness is never low, hits 100\% regularly, and trends upward.  

We can conclude that, at least empirically, deferred even splitting is superior to even splitting for large $r$.  

\phantom{x}

 We see these trends in our simulations.  What about in theory?

\paragraph{Results.} 
Our goal is to find a splitting algorithm for every value of $r$ that achieves fill bounded away from 50\%.
\Cref{tab:fills} and \Cref{fig:results} summarize our results. We are able to prove bounds for different parameter regimes and find that different splitting algorithms are better in different ranges of $r$.
Specifically, our best theoretical results for small $r\leq 7B/18$ alternate between two algorithms: even splitting and deferred even splitting. Then, around the value of $B/2$ where these algorithms only achieve $50\%$ fill, we use a form of \emph{uneven} splitting defined in \Cref{sec:LargeBatchSizes}. Finally, for larger batch sizes $r\geq 2B/3$, we use deferred even splitting.

\paragraph{Remark.} We briefly elaborate on the alternation between even splitting and deferred even splitting for small $r$. For deferred even splitting, for some ranges of $r$ we can use a black box reduction to Yao's setting. This occurs when the divisibility properties of the batch size enable us to treat a chunk of insertions as a single insertion in Yao's setting. However, for much of the parameter range, $r$ does not have such nice divisibility properties. In these cases, we analyze even splitting from scratch by extending the analysis of Yao. This analysis constitutes the technical bulk of the paper.

\begin{table}\label{tab:FillTable}
    \centering
    \begin{tabular}{l|c|c|c|c}
        $r/B$ & Fill & Fill value range &  Splitting Algorithm & Reference\\
        \hline
        \hline
        $\bigl(1/(2i),\,1/(2i - 1)\bigr]$, $i\geq 1$ & $\frac{2ir}{B}\,(H_{2i} - H_{i}) $ & $[0.5,1]$ & Deferred even split & \Cref{Sec:EvenDeferredAnalysis}
        \\[10pt]
         \hline
         $[1/B,0.0058]$& $\ln(2)-\frac{5r}{B}$  & $[0.664,0.693]$ & Even split & \Cref{subsec:LowerBounds}\\[10pt]
         \hline
         
         $(0.0058,0.21]$& $\frac{2\left(B+1\right)}{3B+1+2r}$  & $[0.583,0.664]$ & Even split & \Cref{subsec:LowerBounds}\\[10pt]
         \hline
         
         $(0.21,7/18]$& $0.583$  & $0.583$ & Even split & \Cref{subsec:LowerBounds}\\[10pt]
         \hline
         
         $(7/18,0.5]$& $\frac{3r}{2B}$ & $[0.583,0.750]$ & Uneven split & \Cref{subsec:UnevenSplit}\\[10pt]
         \hline
         $(0.5,2/3]$& $\frac{10r}{9B}$ & $[0.555, 0.740]$ & Uneven split &\Cref{subsec:UnevenSplit} \\[10pt]
         \hline
         $(2/3,1]$& $\frac{r}{B}$ & $[0.666,1]$ & Deferred even split & \Cref{subsec:LargeRegime}
         \\[10pt]
         \hline
         $(1,\infty)$& $\max\left\{\frac{0.5+r/B}{\lceil 1+r/B \rceil },2/3\right\}$ & $[0.666,1]$ & Deferred even split & \Cref{subsec:LargeRegime}
         \\[5pt]
         \hline
    \end{tabular}
    \caption{Fill achieved by our algorithms for different batch sizes}
    \label{tab:fills}
\end{table}

\begin{figure}
    \centering
    \includegraphics[width=0.9\linewidth]{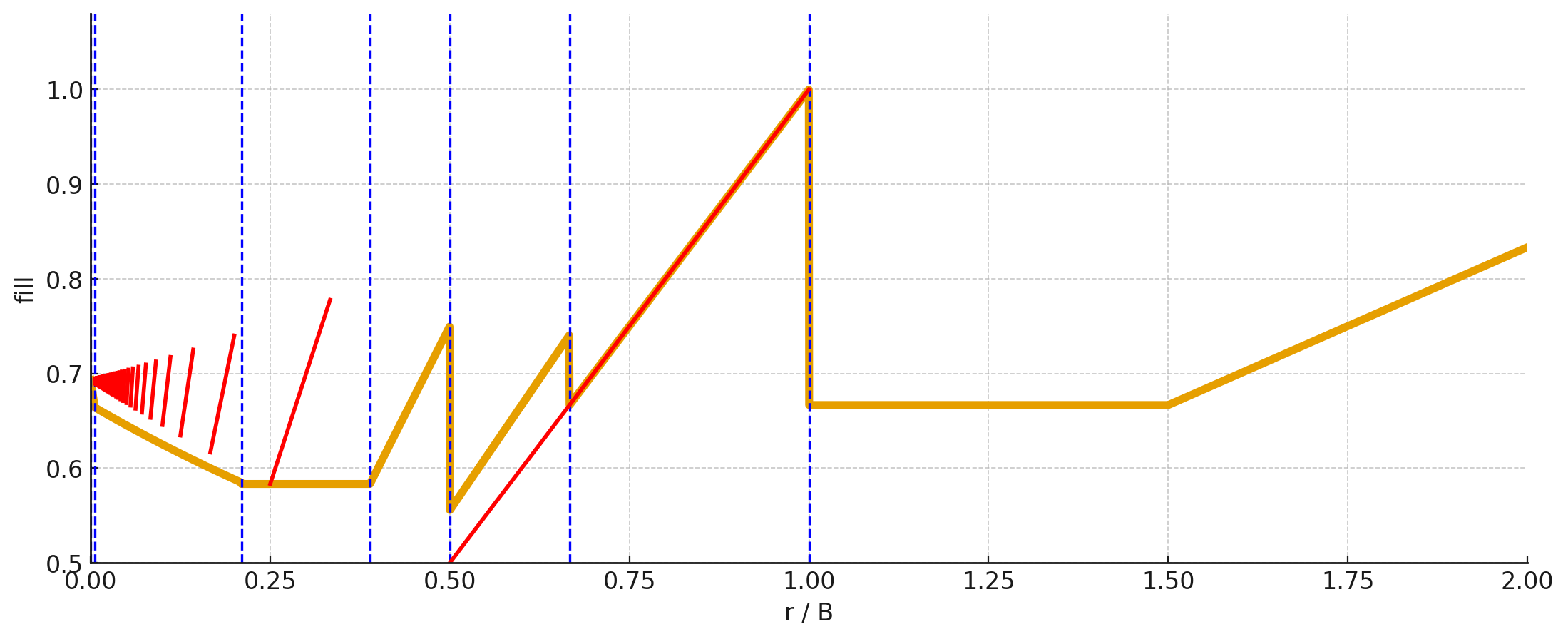}
    \caption{Our theoretical bounds, presented in \Cref{tab:fills}, plotted as a function of $r/B$. The red segments represents the first line in \Cref{tab:fills}, while the continuous orange plot represents the remaining entries.}
    \label{fig:results}
\end{figure}

\paragraph{Open Problems.} Our results leave several open problems:
\begin{enumerate}
    \item Can we get tight bounds for even splitting for $r>1$? 
    \item The exact analytical bounds for deferred even splitting on small $r$ have some gaps.  Can we extend our approach for analyzing even splitting on small $r$ to achieve bounds for deferred even splitting for values of $r$ to fill the gaps between our bounds (the red lines in Figure~\ref{fig:intro-exp-small})? 
    \item Is there a splitting algorithm that achieves greater than 50\% fill without knowing $r$?
    \item Is there a splitting algorithm that achieves greater than 50\% fill when $r$ can vary between insertions?  What about when $r$ is a lower bound on the length of the insertion runs, that is, where the length of the run can vary from insertion run to insertion run, but in each case has at least $r$ items?
\end{enumerate}

\paragraph{Paper organization.}
To achieve our bounds, the main technical lifting is in achieving bounds for even splitting when $r$ is relatively small (rows 2, 3, and 4 of \Cref{tab:fills}). 
For this reason, Sections~\ref{sec:Preliminaries}, ~\ref{sec:technical-overview}, and ~\ref{Sec:EvenSplittingAnalysis} are mostly dedicated to this regime. 
In \Cref{sec:Preliminaries}, we formally define the problem and show how to recast it as analyzing the behavior of a certain recurrence.
In \Cref{sec:technical-overview}, we present a high-level overview of our techniques,
including a discussion of the significant challenges involved in extending Yao's analysis for uniformly random insertions to work for batched random insertions.
In \Cref{Sec:EvenSplittingAnalysis}, we prove our results for the small-$r$ regime. 
In \Cref{subsec:UnevenSplit}, we turn to the medium-$r$ regime and show our analysis of uneven splitting in this regime, which uses the tools we developed for even splitting but is technically simpler. In \Cref{subsec:LargeRegime}, we analyze the regime of large $r$, which is technically simple and does not require any machinery. 
Finally, in \Cref{Sec:EvenDeferredAnalysis}, we present a closed form for the fullness of the deferred even splitting strategy for a range of batch sizes.

\section{Preliminaries} \label{sec:Preliminaries}
As previously discussed, we study the problem of measuring the internal fragmentation of B-trees subject to batches of random insertions under different splitting algorithms.
We use the standard formulation of $B+$ trees \cite{DBLP:books/daglib/0011128} in which all keys are stored in the leaves. This affords a simplified model of the problem in which all keys are stored in sorted order across a collection of memory blocks, and the algorithm is specified by two central choices. 
The first is how to split a memory block. In other words, when an insertion occurs into a block that is full, the algorithm splits it into two memory blocks and chooses how to distribute the keys between the new blocks.
The second is how to break ties. Specifically, when an insertion occurs in the interval between the maximum of one memory block and the minimum of the next, the algorithm needs to decide which memory block to insert the key into.  
Our measure of interest is the average block size in the structure (i.e., total number of keys divided by the number of blocks) as the number of keys inserted into the tree grows.

The rest of this section is primarily dedicated to reviewing the modeling we use for the purpose of analyzing even splitting, as this analysis is the more technically involved part of this work. We formalize and analyze our other splitting procedures, (i.e. uneven and deferred even splitting) in \Cref{sec:LargeBatchSizes}. 

\subsection{Uniformly random insertions, batched random insertions}
In the model studied by Yao~\cite{DBLP:journals/acta/Yao78,DBLP:journals/corr/abs-2202-04185}, at each time step, a key is inserted into a uniformly random position. More precisely, assume the keys inserted thus far into the structure are $k_1,...,k_n$. These keys induce $n+1$ disjoint intervals: 
\[
(-\infty,k_1),(k_1,k_2),...,(k_{n-1},k_n),(k_n,\infty).
\]
At each step, the adversary inserts a key into one of these intervals, where each interval has probability $\frac{1}{n+1}$ of receiving the insertion, independently of any prior insertions.
As discussed in \Cref{sec:intro}, we study a generalization of this workload in which at each step, the adversary does not insert a single key into the interval chosen uniformly at random, but rather $r$ keys provided in ascending order. We proceed to formally describe the even splitting algorithm and to summarize Yao's analysis of even splitting against uniformly random insertions. We then explain how this analysis naturally generalizes to part of our batched random insertion setting.

\paragraph{Even splitting algorithm:}
In an even splitting algorithm, when an insertion occurs into a full block, i.e. one that contained $B$ keys prior to the insertion, we split the block into two new blocks and disperse the keys evenly between the two. More precisely, we place the smallest $\lfloor\frac{B+1}{2}\rfloor$ keys into the left block resulting from the split, and the largest $\lceil\frac{B+1}{2}\rceil$ keys into the right block. For simplicity of presentation, whenever we discuss even splitting, we assume $B$ is odd. For tie breaking, i.e., when an insertion occurs in the interval between the maximum of one memory block and the minimum of the next, we always choose to insert the key into the left block. Note that this choice is arbitrary as the analysis still holds with minor modifications if we make the opposite choice.

\subsection{Even splitting against uniformly random insertions}
\label{subsec:UniformModeling}
We begin by reviewing Yao's analysis for even splitting against uniformly random insertions. To compute the average block size, we must analyze how the composition of blocks changes in our structure. 
Observe that after the first split occurs, no blocks of size less than $d:=\frac{B+1}{2}$ can exist in the data structure for the remainder of time, and thus it suffices to analyze the expected number of blocks of size $i$ for each $i\in\{d,d+1,...,2d-1=B\}$.

For all $n$ and $i\in\left\{ d,d+1,...,2d-1\right\}, $ we let $x_{i}^{n}$ denote the number of blocks in the data structure which contain exactly $i$ keys after a total of $n$ keys have been inserted. We refer to this type of block as an `$i$-block'.

We let $v_{n}$ denote the vector of the expectation values of each of these random variables, i.e., 
\[
v_{n}=\left(\mathbb{E}\left(x_{d}^{n}\right),\mathbb{E}\left(x_{d+1}^{n}\right),...,\mathbb{E}\left(x_{2d-1}^{n}\right),\right)
\]

Each component of this vector can be represented by a scalar recurrence. 
\[
\mathbb{E}\left(x_{i}^{n+1}\right)=\mathbb{E}\left(x_{i}^{n}+\Delta_{i}^{n}\right)=\mathbb{E}\left(x_{i}^{n}\right)+\mathbb{E}\left(\Delta_{i}^{n}\right)
\]
where $\Delta_{i}^{n}$ is the difference in the number of blocks of size $i$ from $v_{n}$ to $v_{n+1}$. 

To completely specify a recurrence we must characterize $\Delta_{i}^{n}$.
Observe that if we assume before the sequence of insertions begins that $-\infty$ is inserted into the structure (i.e. a dummy key smaller than all other keys is inserted), then the probability of hitting a particular block of size $i$ when $n$ keys are present in the structure neatly becomes exactly $\frac{i}{n}$. Thus, if a total of $n$ keys have been inserted (including the dummy key), then with probability exactly $i\cdot \frac{x_{i}^{n}}{n}$ the next insertion occurs in a block of size $i$.
Thus, for every $i>d$, we have:
\[
\mathbb{E}(\Delta_{i}^{n})=-i\cdot \frac{\mathbb E(x_{i}^{n})}{n}+(i-1)\cdot\frac{\mathbb E(x_{i-1}^{n})}{n} 
\]
This is because if the next insertion occurs in a block of size $i$, it then transforms into a block of size $i+1$ (or split into two blocks of size $d$), reducing the number of blocks of size $i$ by 1, while if the insertion occurs in a block of size $i-1$, then the number of blocks of size $i$ increases by 1.
For the case of $i=d$, if the insertion occurs in a block of size $d$, then the number of blocks of size $d$ is decreased by $1$. If the insertion occurs in a block of size $2d-1$, then a split is triggered and the number of blocks of size $d$ increases by $2$. So, we obtain:
\[
\mathbb{E}\left(x_{d}^{n+1}\right)=\mathbb{E}\left(x_{d}^{n}\right)-d\cdot \frac{\mathbb{E}\left(x_{d}^{n}\right)}{n}+2\cdot (2d-1)\cdot \frac{\mathbb{E}\left(x_{2d-1}^{n}\right)}{n} \]
We thus obtain the following recurrence for $v_n$: $\forall n\geq n_{0}:v_{n+1}=\left(I+\frac{1}{n}\cdot A\right)v_{n}$, where:
\[
A=\begin{bmatrix}-d & 0 & \cdots & \cdots & 0 & 2B\\
d & -\left(d+1\right) & \cdots & \cdots & 0 & 0\\
0 & d+1 & \cdots & \cdots & 0 & 0\\
\vdots & \vdots & \ddots & \ddots & \vdots & \vdots\\
0 & 0 & \cdots & \cdots & -\left(B-1\right) & 0\\
0 & 0 & \cdots & \cdots & B-1 & -B
\end{bmatrix}
\]

\subsection{Modeling for even splitting subject to batched insertions}\label{subsec:ModelingBatched}
We now consider the case of batched random insertions. As mentioned above, at each time step the adversary inserts a batch of $1\leq r<\frac{B}{2}$ consecutive keys into a randomly sampled interval.
We demonstrate how the modeling method reviewed in \Cref{subsec:UniformModeling} naturally generalizes to this case.
Like \Cref{subsec:UniformModeling}, we consider the vector describing the expected composition of our data structure, or more precisely the expected quantities for each type of block:
\[
v_{n}=\left(\mathbb{E}\left(x_{d}^{n}\right),\mathbb{E}\left(x_{d+1}^{n}\right),...,\mathbb{E}\left(x_{2d-1}^{n}\right)\right).
\]
We can similarly obtain scalar recurrence for each component, but it is convenient for analysis to take a step size of $r$:
$\mathbb{E}\left(x_{i}^{n+r}\right)=\mathbb{E}\left(x_{i}^{n}\right)+\mathbb{E}\left(\Delta_{i}^{n}\right)$, where now $\Delta_{i}^{n}$ denotes the expected change in the quantity of blocks of size $i$ between $v_n$ and $v_{n+r}$. For $i>d$, if the first key in a batch is inserted into a block of size $i$, then by the end of the batch the number of blocks of size $i$ has decreased by $1$, while if instead it hits a block of size $d+(i-d-r)\bmod d$, then by the end of the batch the number of blocks of size $i$ has increased by $1$. 
We can therefore obtain, for every $i>d$, and letting $w_i$ denote $d+(i-d-r) \bmod d) $:
\[
\mathbb{E}(\Delta_{i}^{n})=-i\cdot \frac{\mathbb E(x_{i}^{n})}{n}+w_i\cdot\frac{\mathbb E(x_{w_i}^{n})}{n}. 
\]
For the case of $i=d$ - if the batch begins in a block of size $d$, then the number of blocks of size $d$ is decreased by $1$. If the batch begins in a block of size $B+1-r$, then the last key in the batch triggers a split, and thus the number of blocks of size $d$ is increased by $2$. If the batch begins in a block of size $k>B+1-r$, then a split occurs during the batch, and following the split additional keys are inserted into one of the created siblings. Thus, the number of blocks of size $d$ is increased by $1$. We therefore obtain:
\begin{equation}\label{eq:FirstComponentRecurrence}
    \mathbb{E}(\Delta_{d}^{n})=-d\cdot \frac{\mathbb E(x_{d}^{n})}{n}+2\cdot (B+1-r)\cdot\frac{\mathbb E(x_{B+1-r}^{n})}{n}
+\sum_{B+1-r<k\leq B}k\cdot \frac{\mathbb E(x_{k}^{n})}{n}.
\end{equation}

We can therefore obtain a matrix recurrence of a similar form to \Cref{subsec:UniformModeling}: $\forall n\geq n_{0}:v_{n+r}=\left(I+\frac{1}{n}\cdot A(B,r)\right)v_{n}$.
We use $A(B,r)$ to denote the transition matrix as it reflects blocks of capacity $B$ and batches of size $r$ (observe that we've specified $A(B,1)$ in \Cref{subsec:UniformModeling}).
We now turn to fully specify $A(B,r)$, based on the scalar recurrences above. Before we do so, a non-standard indexing method that will be convenient to use for later analysis.

\subsubsection{Indexing method}
 As mentioned, following the first split, all blocks in our structure have size at least $d$, which is why it suffices to track blocks of at least this size. This is reflected in the fact, that the first entry in $v_n$ corresponds to the number of blocks of size $d$, and the last entry corresponds to the number of blocks of size $B$. Similarly, the first row in $A(B,1)$ (specified in \Cref{subsec:UniformModeling}) corresponds to the recurrence relation for expected the number of blocks of size $d$.
It is significantly convenient for the purpose of later analysis to use indexing that reflects this - i.e to index entries in $d$-dimensional vectors, and $d\times d$-dimensional matrices by the block sizes they corresponds to. 
\\ To describe this formally: We first define $f:\left\{ 0,...,d-1\right\} \to\left\{ d,...,2d-1\right\} $ by $f\left(i\right)=d+i$, note that $f$ is a bijective mapping.

Let $B$ be any $d\times d$ matrix, and $w$ be a $d$-dimensional vector. We will use $f\left(i\right)$ to index the $i^{th}$ row/column in $B$ and the $i^{th}$ component of $w$. For example, $B_{i,j}$ refers to $B_{i-d,j-d}$ in standard indexing.

For a principle submatrix $P$ of a $d\times d$ matrix $B$, we will use $f\left(i\right)$ to index the $j^{th}$ row/column in $P$, where $i$ is the corresponding index to $j$ in the greater matrix $B$. Similarly, for $u$, a subvector of a $d$-dimensional vector $w$, we will use $f\left(i\right)$ to index the $j^{th}$ entry, where $i$ is the corresponding index to $j$ in $w$.

\subsubsection{Specifying the general transition matrix}
We now specify the general transition matrix $A(B,r)$, using the scalar recurrences presented in \Cref{subsec:ModelingBatched} and our non-standard indexing method. The first row needs to be presented separately, as it has a unique structure. Further, it is convenient to split the remaining rows into two categories so we can explicitly specify the entries without using the modulo operation. Note that all entries not explicitly mentioned below are zero. Below, $A$ always denotes $A(B,r)$; recall that we are assuming that $r < B/2$.

\textbf{First row: } 
\begin{align*}
&A_{d,d}=-d,\\& A_{d,B+1-r}=2\cdot\left(B+1-r\right),\\ &\text{for every } B+1-r<k\leq B:A_{d,k}=k.
\end{align*}

\textbf{Wraparound rows}: \begin{align*}
\forall d&<k\leq d+r-1:\\&A_{k,k}=-k,\\&A_{k,d+k-r}=d+k-r\,\, 
\end{align*}

\textbf{Remaining rows: }
\begin{align*}
\forall d+r&\leq k\leq B:\\&A_{k,k}=-k,\\&A_{k,k-r}=k-r.
\end{align*}

Observe that $A_{i,j}$ can be described as follows: $\frac{A_{i,j}}{j}$ is the signed difference that occurs in the number of blocks of size $i$, if a batch of insertions begins in a block of size $j$.

 For a concrete example that shows how we obtain this matrix directly from the recurrence relations we specified, we demonstrate how the first row for $A$ is obtained.
As mentioned in \Cref{subsec:ModelingBatched}:
\[
\mathbb{E}\left(x_{d}^{n+r}\right)=\mathbb{E}\left(x_{d}^{n}\right)+\mathbb{E}\left(\Delta_{i}^{n}\right)
\]
Therefore, if we denote the first row of $A$ by $\vec{a}$, for our desired recurrence to be fulfilled, it needs to satisfy  (letting $\langle,\rangle$ denote the standard inner product):
\begin{equation}\label{eq:requirements}
    \langle e_1+\frac{1}{n}\cdot \vec a\space , v_n\rangle =\mathbb{E}\left(x_{d}^{n}\right)+\mathbb{E}\left(\Delta_{d}^{n}\right)\iff \langle \vec a\cdot v_n\rangle =n\cdot \mathbb{E}\left(\Delta_{d}^{n}\right)
\end{equation}
Now, by \Cref{eq:FirstComponentRecurrence}:
\[
n\cdot \mathbb{E}\left(\Delta_{d}^{n}\right)=\mathbb{E}(\Delta_{d}^{n})=-d\cdot \mathbb E(x_{d}^{n})+2\cdot (B+1-r)\cdot\mathbb E(x_{B+1-r}^{n})
+\sum_{B+1-r<k\leq B}k\cdot \mathbb E(x_{k}^{n})=\langle u,v_n \rangle
\]
where $u=(-d,0,...,0,2(B+1-r),B+2-r,... B)$.
So, the first row of $A$, $\vec a$, needs to satisfy:
\[
\langle \vec a,v_n\rangle= \langle u,v_n\rangle.  
\]
So, setting $\vec a=u$ satisfies \Cref{eq:requirements}.

\section{Technical Overview}
\label{sec:technical-overview}
We begin by providing a technical overview of our contributions. We spend the majority of this section giving a technical overview of the analysis for even splitting in B-trees against uniformly random batched insertions, as it is the most technically challenging. As outlined in \Cref{sec:Preliminaries}, the recurrence describing the behavior of even splitting subject to batched random insertions of size $r$ is of the form $v_{n+r}=\left(I+\frac{1}{n}\cdot A(B,r)\right)v_{n}$, and $A(B,r)$ is termed the \defn{transition matrix}.
For $r=1$ we get the exact recurrence analyzed by Yao \cite{DBLP:journals/acta/Yao78}.

\paragraph{Main objective}
Our goal is to compute the asymptotic average block size in the data structure. Observe that in fact $v_n/n$ can be used to characterize the average block size (see \Cref{subsec:AverageBlockSize} for more details). This is intuitive since the ratio between different entries in $v_n/n$ tells us the ratio between the quantities of blocks of different sizes in our structure. So, our goal from now on is to analyze $\lim _{n\to\infty }v_n/n$.
It is not hard to verify that if $v_n/n$ admitted a steady state (i.e. $v_{n+r}/(n+r)=v_n/n$) then that steady state must be an eigenvector of $A(B,r)$ corresponding to eigenvalue $r$. This makes it intuitive to expect that $v_n/n$ should converge to an eigenvector of $A(B,r)$. The key technical challenges have to do with proving this intuition, and as we'll see, analyzing the structure of an eigenvector of $A(B,r)$.

\paragraph{Yao's apparoach}
The analysis of random B-trees provided in \cite{DBLP:journals/acta/Yao78} can be interpreted as consisting of two primary stages:
\begin{enumerate}
    \item The first step involves proving that $\frac{v_n}{n}$ converges to an eigenvector of $A(B,1):=A$. As mentioned, this is valuable as $\frac{v_n}{n}$ can be used to represent the expected average block size.
    \item Then, the eigenvector is computed exactly and is used to obtain the average block size.
\end{enumerate}
\paragraph{Generalizing Yao's apparoach}
Much like Yao, our general approach is to prove that each recurrence $v_{n+r}=\left(I+\frac{1}{n}\cdot A(B,r)\right)v_{n}$ converges to an eigenvector $A(B,r)$, and then rely on analyzing this eigenvector. However, applying Yao's methods directly is difficult and we thus deviate significantly in how we go about these tasks. 
The difficulties in generalizing from Yao's analysis are as follows:

\phantom{x}

  \underline{First Challenge:}
 Proving that $v_n/n$ converges to an eigenvector of $A(B,r)$ for every $r$ is a significantly harder task. Yao's approach for $A = A(B,1)$ relied on a low-level analysis that explicitly computed very specific information about this particular matrix $A$; for example, his analysis required computing the characteristic polynomial of $A$ and factoring it according to a specific form. But the moment we turn to $r > 1$, such an approach seems quite difficult to execute for two reasons.

    Firstly, even if we could repeat Yao's factorization of the characteristic polynomial for a specific value of $r$, it is unclear that there should exist a general formula for the characteristic polynomial that can be parameterized on $r$. As suggested by \Cref{fig:intro-exp-small}, the characteristics of $A(B,r)$ seem to vary in a non-monotonic way with $r$, so there is no reason to assume a single formulation can work for all values of $r$ simultaneously. 
    
    Secondly, even if we could achieve a general formulation for the characteristic polynomial and obtain complete spectral characterization of transition matrix for general $r$, arguing about convergence for general $r>1$ requires different techniques from $r=1$. Specifically, $A(B,1)$ is a diagonalizable matrix, which affords decomposing the initial state into a basis of eigenvectors and computing limits on each component separately. $A(B,r)$ on the other hand is generally not diagonalizable; for example $A(15,4)$ is not a diagonalizable matrix.
    
   \phantom{x} 
   
   \underline{Second Challenge:} Even if we know the (normalized) recurrence converges to an eigenvector, the analysis of this eigenvector in order to derive bounds is much more involved when $r>1$. For $r=1$, the eigenvector admits a neat closed form, and thus the asymptotic average block size can be computed precisely.
    By contrast, for $r>1$ we encounter two similar difficulties to the previous step. 
    Firstly, considering the non-smooth nature of the curve in \Cref{fig:intro-exp-small} and the fact that the eigenvector can be used to express the average block size (see \Cref{subsec:AverageBlockSize} for details), we should not expect a closed form that can work for every value of $r$ simultaneously. Secondly, even if we set a specific $r>1$ and try to characterize the eigenvector for the particular $A(B,r)$, the system of equations for the eigenvector that is generated by $A(B,r)$ is significantly more complicated than that of $A(B,1)$ due to its different structure, and there is thus no reason to assume it admits a neat closed form (as far as we can tell it does not).

\paragraph{Overcoming the challenges for generalization}

 \underline{First Challenge:} To prove $v_n/n$ converges, we first avoid any precise analysis of the characteristic polynomial, and instead rely on simple structural characteristics that hold for all $A(B,r)$ in order to characterize their spectrum. Then, we use this characterization to argue about convergence through the complex Jordan form of $A(B,r)$. See \Cref{sec:ConverganceTheorem} for the details.

\phantom{x}

 \underline{Second Challenge:} We avoid the need to compute the eigenvector precisely, and provide looser characterizations in order to derive lower bounds. We observe that while the eigenvector does not admit a neat closed form, certain subvectors of the eigenvector are much more structured (although the way in which the values of these subvectors relate to each other is unclear). We then analyze these subvectors in order to provide lower bounds for the average block size. See \Cref{subsec:LowerBounds} for details.

\subsection{Algorithms for larger $r$} \label{prelims:LargeBatches}
The discussion so far has pertained to the even splitting algorithm, which produces favorable bounds for $r<0.388B$ (see \Cref{tab:FillTable} for more details). 
As discussed in the introduction, our bounds for even splitting decay as $r$ approaches $B/2$, and we thus appeal to a different set of algorithms for larger values of $r$. These algorithms are technically simpler to analyze; we now provide a brief overview.

\paragraph{{\boldmath Moderately large $r$: $\frac{7/18} B \leq r \leq \frac{2}{3}B$}}

 In this regime, we obtain favorable bounds by using uneven splitting algorithms (see \Cref{sec:LargeBatchSizes}).
These algorithms are easier to analyze because choosing specific appropriate splits, we are able to confine the possible sizes of blocks our structure contains at the end of each insertion batch to a small predetermined set. In particular, we are able to ensure that there are at most three possible blocks sizes; for example, for $r = B/2$, our algorithms will ensure that all blocks have size $B/4, B/2$ or $3B/4$. We can then readily analyze the recurrence using tools we developed for the analysis of even splitting. These tools are simple to apply because our transition matrix is at most $3\times 3$, as we have at most $3$ different possible block sizes.
See \Cref{sec:LargeBatchSizes} for more details.

\paragraph{{\boldmath Large values of $r$: $r > \frac{2}{3}B$}}

 This is by far the easiest case, No linear algebra is required to conclude that deferred even splitting ensures good space utilization. See \Cref{sec:LargeBatchSizes} for details.

\section{Analysis of even splitting against batched insertions}
\label{Sec:EvenSplittingAnalysis}
Our goal in this section is to establish the bounds presented in \Cref{tab:fills} for even splitting.
The organization of the analysis is as follows: \Cref{subsec:SpectralCharacterization} provides a spectral characterization of our transition matrices, $A(B,r)$ (see \Cref{sec:Preliminaries} for a specification of these matrices). Then, in \Cref{sec:ConverganceTheorem}, we use this characterization to prove the normalized recurrence relations presented in \Cref{sec:Preliminaries} converge to an eigenvector of the transition matrix. \Cref{subsec:AverageBlockSize} demonstrates how any eigenvector of the transition matrix can be used to compute the limit of the average block size. Then, we spend \Cref{subsec:LowerBounds} analyzing the structure of an eigenvector to derive lower bounds for the asymptotic average block size.
\subsection{Spectral characterization of the transition matrix} \label{subsec:SpectralCharacterization}
\subsubsection{Spectral properties of irreducible Metzler matrices}
Our analysis uses a well known characterization of the spectrum of irreducible Metzler matrices. For completeness, we begin with a self-contained review of the definition of Metzler matrices and a standard proof. For further reading, refer to \cite{0521386322}.
\begin{definition}

A square matrix $A\in\mathbb{R}^{d\times d}$ is a Metzler
matrix, if and only if all off-diagonal entries in $A$ are non-negative.
 
\end{definition}
\begin{definition}

For a square matrix $A\in\mathbb{R}^{d\times d}$, we define the underlying
directed graph of $A$, denoted by $G\left[A\right]$, as a directed
graph on $V=\left\{ 0,...,d-1\right\} $ in which for every $i,j\in\left\{ 0,...,d-1\right\} $,
there exists a directed edge from $i$ to $j$ in $G\left[A\right]$
if and only if $A_{i,j}\neq0$.

\end{definition}
\begin{definition}

We say that a square matrix $A$ is irreducible, if $G\left[A\right]$
is a strongly connected graph.

\end{definition}
\begin{lemma}\label{lem:MetzlerMatrices}

Let $A\in\mathbb{R}^{d\times d}$ be an irreducible Metzler matrix,
and assume there exists a positive (entry-wise) $w$ such that $w^{T}\cdot A=r\cdot w^{T}$
for some $r>0$, then: 

\begin{enumerate}
    \item $r$ is a simple eigenvalue of $A$, and there exists a positive eigenvector corresponding to it.
    \item For every other eigenvalue $\lambda\in\sigma\left(A\right)$, we have $\mathcal{R}e\left(\lambda\right)<r$, where $\mathcal{R}e(\cdot)$ denotes the real part of a complex number.
\end{enumerate}
\end{lemma}
\begin{proof}

Let $A\in\mathbb{R}^{d\times d}$ be irreducible Metzler, and assume
$w^{T}\cdot A=r\cdot w^{T}$ for a positive $w$. Note that this already
means $r$ is an eigenvalue of $A$. 
Now, let $B=A+c\cdot I$ where
\[
c=max\left\{ \left|A_{i,i}\right|:i\in\left\{ 0,...,d-1\right\} \right\}
\]

Observe that $B$ is a non-negative, irreducible matrix. Further:
\[
w^{T}\cdot B=w^{T}\cdot\left(A+c\cdot I\right)=r\cdot w^{T}+c\cdot w^{T}=\left(r+c\right)\cdot w^{T}
\]
So, $w^{T}$ is a positive left-eigenvector of $B$ corresponding
to eigenvalue $r^{\prime}=r+c>0$. By applying Perron-Frobenius theorems (\cite{0521386322} for further reading) for non-negative irreducible matrices we can conclude the following: 

1. $r^{\prime}$ is the spectral radius of $B$, this implies that
for any other eigenvalue $\lambda$ of $B$, we have $\mathcal{R}e\left(\lambda\right)<r^{\prime}$.

2. $r^{\prime}$ is a simple eigenvalue, and there exists a positive
eigenvector of $B$ corresponding to it.

 Now, since the eigenspace of $A$ corresponding to $r$ is equal to
the eigenspace of $B$ corresponding to $r^{\prime}$, we can conclude
$r$ is simple, and that there exists a positive eigenvector of $A$
corresponding to it. Further, let $\lambda\in\sigma\left(A\right)$,
then $\lambda+c\in\sigma\left(B\right)$, and $\mathcal{R}e\left(\lambda+c\right)<r^{\prime}\implies\mathcal{R}e\left(\lambda\right)+c<r+c\implies\mathcal{R}e\left(\lambda\right)<r$,
as needed.
\end{proof}
\subsubsection{Analyzing the spectral properties of $A(B,r)$}
\label{subsec:MatrixCharacterization}

For every $k\in\left\{ d,...,2d-1\right\} $, let $\Delta_k\in\left\{ -1,0,1,2\right\} ^{d}$ be the vector describing the change seen in $v_{n}$ when a batch of insertions hits a block of size $k$, in other words, $\Delta_{k}=v_{n+r}-v_{n}$ (coordinate-wise subtraction) provided that the next batch occurs on a block of size $k$.
As we prove next, for each $k$, the $k$'th column of the transition matrix is equal to $\Delta _k\cdot k$. To illustrate why this is intuitive - Observe that the $k^{th}$ column has values in $\left\{ -1,0,1,2\right\} $ when normalized by $k$. $A_{k,k}=-1$, which corresponds to decreasing the number of blocks of size $k$ when we hit one. Further, for all the remaining entries, $A_{i,k}\neq0$ if and only if hitting a block of size $k$ increases the number of blocks of size $i$ by $1$ or $2$. This is due to our basic modeling; the $i^{th}$ row reflects how the number of blocks of size $i$ changes. We proceed to state and prove this formally. 

\begin{lemma}\label{lem:ColumnsOfTransitionMatrix}

The $k^{th}$ column of $A\left(B,r\right)$ is equal to $\Delta{_k}\cdot k$.

\end{lemma}
\begin{proof}

For $d\leq k<2d-r$, we have $A_{k,k}=-k$ and $a_{k+r,k}=k$, and
indeed $\Delta_{k}$ describes a reduction in the number of blocks
of size $k$ and an increase in the number of blocks of size $k+r$.

For $k=2d-r$ we have $A_{k,k}=-k$, and $A_{d,k}=2\cdot k$, and
indeed $\Delta_{k}$ describes a reduction in the number of blocks
of size $k$ and an increase by $2$ in the number of blocks of size
$d$.

For $2d-r<k\leq2d-1$ we have $A_{k,k}=-k$, $A_{d,k}=k$, and $A_{k+r-\left(B+1\right)+d,k}=k$,
and indeed $\Delta_{k}$ describes a reduction in the number of blocks
of size $k$ and an increase by $1$ in both the number of blocks
of size $d$, and the number of blocks of size $k+r-\left(B+1\right)+d$. 
\end{proof}
\begin{lemma}\label{lem:eigenvalueOfTransitionMatrix}

$r$ is an eigenvalue
of the transition matrix $A(B,r)$, and the matrix has a positive
(entry-wise) left-eigenvector corresponding to $r$. 
\end{lemma}

\begin{proof}
 We demonstrate a positive left eigenvector corresponding to eigenvalue
$r$: let $w=\left(d,...,2d-1\right)$, note that by
our indexing method $w_{k}=k$. By \Cref{lem:ColumnsOfTransitionMatrix},$\left[w^{T}\cdot A\left(B,r\right)\right]_{k}$
(where $\left[\cdot\right]_{k}$ is used to denote the $k^{th}$ column)
is equal to $w^{T}\cdot k\cdot\Delta_{k}=k\cdot\left(w^{T}\cdot\Delta_{k}\right)$. Observe that $\text{\ensuremath{w^{T}\cdot\Delta{k}}}$ is exactly
equal to the net change in the number of elements when inserting $r$
elements if a block of size $k$ is hit, which is invariably equal
to $r$, so overall: $\left[w^{T}\cdot A\left(B,r\right)\right]_{k}=k\cdot r=r\cdot\left[w^{T}\right]_{k}$.
\end{proof}

\begin{definition}

Let $\left\langle r\right\rangle $ be the cyclic subgroup of $\mathbb{Z}_{d}$
generated by $r$, and we let $S$ denote $d+\left\langle r\right\rangle =\left\{ d+x:x\in\left\langle r\right\rangle \right\} $.

\end{definition}
\begin{lemma}

$S$ is a cycle in $G\left[A\left(B,r\right)\right]$.

\end{lemma}
\begin{proof}
 It suffices to demonstrate $S$ is a cycle in $G\left[A\left(B,r\right)^{T}\right]$.
For this purpose, it would suffice to show that for each $d+x\in S$,
there exists an edge from $d+x$ to $d+\left(x+r\right)\bmod d$
in $G\left[A\left(B,r\right)^{T}\right]$. Let $k=d+x\in S$, we know
that the $k^{th}$ column in $A\left(B,r\right)$ is $\Delta^{k}$,
i.e. the change we see when a batch of insertions hits a block of size
$k$. Hitting a block of size $k$ increases the number of blocks
of size $d+\left(x+r\right)\bmod d\in S$, and thus the
$\left(d+\left(x+r\right)\bmod d\right)^{th}$ component of
the $k^{th}$ column is non-zero, and therefore we have an edge from
$d+x$ to $d+\left(x+r\right)\bmod d$ in $G\left[A\left(B,r\right)^{T}\right]$,
as needed.
\end{proof}
\begin{corollary}\label{cor:strongConnectivity}

Let $A\left(B,r\right)_{S}$ be the principal submatrix of $A\left(B,r\right)$
corresponding to $S$. Then $G\left[A\left(B,r\right)_{S}\right]$
is strongly connected.

\end{corollary}

This follows from the fact that $G\left[A\left(B,r\right)_{S}\right]$ is the induced subgraph of $G\left[A\left(B,r\right)\right]$ on the vertex set $S$, and we've proven that $G\left[A\left(B,r\right)\right]$ has a cycle on vertex set $S$. Therefore $G\left[A\left(B,r\right)_{S}\right]$ is strongly connected (however, in general it is not simply a cycle graph). 

\begin{lemma}

$r\in\sigma\left(A\left(B,r\right)_{S}\right)$, and further, $A\left(B,r\right)_{S}$
has a positive (entry-wise) left eigenvector corresponding to $r$.

\end{lemma}
\begin{proof}

We know that $\Delta^{k}$ reflects the change in blocks we see when
an insertion hits a block of size $k$. Observe that when a block
of size $k\in S$ is hit, we only see changes in the quantities of
blocks which sizes are in $S$. More specifically, if $k=d+x\in S$,
the only non-zero entries are $k$, $d+\left(x+r\right)\bmod d\in S$,
and possibly $d\in S$. Overall in this case $\Delta^{k}$ only has
support on $S$.

Therefore, for every $k\in S$, and for $w=\left(d,...,2d-1\right)$,
we have:

\[
\left(w_{S}\right)^{T}\left[A\left(B,r\right)_{S}\right]_{k}=w^{T}\left[A\left(B,r\right)\right]_{k}
\] 
Now, according to the proof of \Cref{lem:eigenvalueOfTransitionMatrix} $w$ is a left-eigenvector of $A(B,r)$, thus:
\[
\left(w_{S}\right)^{T}\left[A\left(B,r\right)_{S}\right]_{k}=r\cdot\left(w^{T}\right)_{k}=r\cdot\left(\left(w_{S}\right)^{T}\right)_{k} 
\]

 Overall, $\left(w_{S}\right)^{T}$ is a positive left eigenvector
of $A\left(B,r\right)_{S}$ corresponding to eigenvalue $r$, as needed.
\end{proof}

\begin{lemma} \label{lem:FinalSpectralCharacterization}

$r$ is a simple eigenvalue of $A\left(B,r\right)_{S}$, for which
there exists a positive (right) eigenvector, and for every other eigenvalue
$\lambda$ we have $\mathcal{R}e\left(\lambda\right)<r$.

\end{lemma}
\begin{proof}
 As a principal submatrix of a Metzler matrix, $A\left(B,r\right)_{S}$
is a Metzler matrix. According to \Cref{cor:strongConnectivity}, $G\left[A\left(B,r\right)_{S}\right]$
is strongly connected, and therefore $A\left(B,r\right)_{S}$ is irreducible
Metzler. We've demonstrated an entry-wise positive $w$ such that
$w^{T}\cdot\left(A\left(B,r\right)_{S}\right)=r\cdot w^{T}$. Therefore
by \Cref{lem:MetzlerMatrices} we can conclude
$r$ is simple, has a positive eigenvector, and for every other
eigenvalue $\lambda$ we have $\mathcal{R}e\left(\lambda\right)<r$.
\end{proof}

\subsection{Proving our recurrences converge} 
\label{sec:ConverganceTheorem}

As presented in \Cref{sec:Preliminaries}, we model the behavior of even splitting through recurrences of the form $v_{n+r}=\left(I+\frac{1}{n}\cdot A(B,r)\right)\cdot v_{n}$. Having provided spectral characterizations for $A(B,r)$ in \Cref{subsec:SpectralCharacterization}, we spend this section proving that under these characterization our  recurrence relations converge (with some normalization).

\begin{lemma} \label{lem:boundForConverganceI}
Let $n_{0},r$ be positive integers, and let $\lambda\in\mathbb{C}$. Then: 
\[
\left|\prod_{k=\lceil\left|\lambda\right|\rceil}^{m-1}\left(1+\frac{\lambda}{n_{0}+k\cdot r}\right)\right|\leq m^{\frac{\mathcal{R}e\left(\lambda\right)}{r}}\cdot\alpha_{\lambda,n_{0},r}
\]
 where $\alpha_{\lambda,n_{0},r}$ is a constant depending only on $\lambda,n_{0},r$.    
\end{lemma}
\begin{proof}
Observe that for any complex $t$ we have: 
\[
\left|1+t\right|=\sqrt{1+2\mathcal{R}e\left(t\right)+\left|t\right|^{2}}\leq\sqrt{\exp\left(2\mathcal{R}e\left(t\right)+\left|t\right|^{2}\right)}\leq \exp\left(\mathcal{R}e\left(t\right)+\frac{1}{2}\cdot\left|t^{2}\right|\right).
\]
where $\exp(x):=e^x$ and $\mathcal{R}e(\lambda)$ denotes the real part of $\lambda$.

Therefore, for $t=\frac{\lambda}{n_{0}+k\cdot r}$ we have:
\begin{equation} \label{inq:1}
\ln\left(\left|\left(1+\frac{\lambda}{n_{0}+k\cdot r}\right)\right|\right)\leq\ln\left(\exp\left(\mathcal{R}e\left(t\right)+\frac{1}{2}\cdot\left|t^{2}\right|\right)\right)    
\end{equation}

Now observe: 
\[
\left|\prod_{k=\lceil\left|\lambda\right|\rceil}^{m-1}\left(1+\frac{\lambda}{n_{0}+k\cdot r}\right)\right|=\prod_{k=\lceil\left|\lambda\right|\rceil}^{m-1}\left|\left(1+\frac{\lambda}{n_{0}+k\cdot r}\right)\right|=\exp\left(\sum_{k=\lceil\left|\lambda\right|\rceil}^{m-1}\ln\left(\left|\left(1+\frac{\lambda}{n_{0}+k\cdot r}\right)\right|\right)\right)
\]
Therefore, according to \Cref{inq:1}, we get:
\[
\left|\prod_{k=\lceil\left|\lambda\right|\rceil}^{m-1}\left(1+\frac{\lambda}{n_{0}+k\cdot r}\right)\right|\leq  \exp\left(\sum_{k=\lceil\left|\lambda\right|\rceil}^{m-1}\mathcal{R}e\left(t\right)+\frac{1}{2}\cdot\sum_{k=1}^{m-1}\left|t^{2}\right|\right)
\]

We bound both sums in the exponent separately, begining with $\sum_{k=\lceil\left|\lambda\right|\rceil}^{m-1}\mathcal{R}e\left(t\right)$:

\[
\sum_{k=\lceil\left|\lambda\right|\rceil}^{m-1}\mathcal{R}e\left(t\right)=\mathcal{R}e\left(\lambda\right)\cdot\sum_{k=\lceil\left|\lambda\right|\rceil}^{m-1}\frac{1}{n_{0}+k\cdot r}:=\mathcal{R}e\left(\lambda\right)\cdot \varphi\left(m\right).
\]
\[
\int_{\lceil\left|\lambda\right|\rceil}^{m}\frac{dx}{n_{0}+x\cdot r}\leq \varphi\left(m\right)\leq\frac{1}{n_{0}+L\cdot r}+\int_{\lceil\left|\lambda\right|\rceil}^{m-1}\frac{dx}{n_{0}+x\cdot r}.
\]

To evaluate the lower bound:

\[
\int_{\lceil\left|\lambda\right|\rceil}^{m}\frac{dx}{n_{0}+x\cdot r}=\frac{1}{r}\cdot\ln\left(\frac{n_{0}+m\cdot r}{n_{0}+\lceil\left|\lambda\right|\rceil\cdot r}\right)\geq\frac{1}{r}\cdot\ln\left(m\right)+\gamma_{n_{0},\lambda,r}.
\]
 where $\gamma_{n_{0},\lambda,r}$ only depends on $n_{0},\lambda,r$.

To evaluate the upper bound:

\[
\frac{1}{n_{0}+L\cdot r}+\int_{\lceil\left|\lambda\right|\rceil}^{m-1}\frac{dx}{n_{0}+x\cdot r}=\frac{1}{n_{0}+\lceil\left|\lambda\right|\rceil\cdot r}+\frac{1}{r}\cdot\ln\left(\frac{n_{0}+\left(m-1\right)\cdot r}{n_{0}+\lceil\left|\lambda\right|\rceil\cdot r}\right)\leq\frac{1}{r}\cdot\ln\left(m\right)+\gamma_{n_{0},\lambda,r}^{\prime}.
\]
where $\gamma_{n_{0},\lambda,r}^{\prime}$ only depends on $n_{0},\lambda,r$.

This implies:
\[
\frac{1}{r}\cdot\ln\left(m\right)+\gamma_{n_{0},\lambda,r}\leq \varphi\left(m\right)\leq\frac{1}{r}\cdot\ln\left(m\right)+\gamma_{n_{0},\lambda,r}^{\prime}.
\]
So, overall we have:

\begin{equation}\label{eq:firstBound}
\sum_{k=\lceil\left|\lambda\right|\rceil}^{m-1}\mathcal{R}e\left(t\right)=\mathcal{R}e\left(\lambda\right)\cdot \varphi\left(m\right)\leq\mathcal{R}e\left(\lambda\right)\left(\frac{1}{r}\cdot\ln\left(m\right)+\beta_{n_{0},\lambda,r}\right)    
\end{equation}
where
\[
\beta_{n_{0},\lambda,r}=\begin{cases}
\gamma_{n_{0},\lambda,r}^{\prime} & \mathcal{R}e\left(\lambda\right)\geq0\\
\gamma_{n_{0},\lambda,r} & \mathcal{R}e\left(\lambda\right)<0.
\end{cases}
\]

Note $\beta_{n_{0},\lambda,r}$ only depends on $n_{0},\lambda$ and $r$ and is in particular independent of $m$.

Now to bound $\frac{1}{2}\cdot\sum_{k=1}^{m-1}\left|t^{2}\right|$:

\begin{equation}\label{secondBound}
\frac{1}{2}\cdot\sum_{k=\lceil\left|\lambda\right|\rceil}^{m-1}\left|t^{2}\right|=\frac{1}{2}\cdot\left|\lambda\right|^{2}\cdot\sum_{k=\lceil\left|\lambda\right|\rceil}^{m-1}\frac{1}{\left(n_{0}+k\cdot r\right)^{2}}\leq\frac{1}{2}\cdot\left|\lambda\right|^{2}\cdot\sum_{i=1}^{\infty}\frac{1}{i^{2}}\leq\left|\lambda\right|^{2}.    
\end{equation}

Overall, applying \Cref{eq:firstBound} and \Cref{secondBound} we achieve:

\[ \exp\left(\sum_{k=\lceil\left|\lambda\right|\rceil}^{m-1}\mathcal{R}e\left(t\right)+\frac{1}{2}\cdot\sum_{k=\lceil\left|\lambda\right|\rceil}^{m-1}\left|t^{2}\right|\right)\leq \exp\left(\frac{\mathcal{R}e\left(\lambda\right)}{r}\cdot\ln\left(m\right)+\mathcal{R}e\left(\lambda\right)\cdot\beta_{n_{0},\lambda,r}+\left|\lambda\right|^{2}\right)=m^{\frac{\mathcal{R}e\left(\lambda\right)}{r}}\cdot\alpha_{n_{0},\lambda,r}.\]
\end{proof}

Lemma~\ref{lem:boundForConverganceI} tells us that, up to a constant, the contribution of an eigenvalue $\lambda$ behaves like a power of $m$ with exponent $\mathcal{R}e(\lambda)/r$. To handle full Jordan blocks, we next need a companion bound for the nilpotent part, which turns out to contribute only polylogarithmic growth.

\begin{lemma} \label{lem:boundForConverganceII}
Let $n_{0},r$ be positive integers, let $\lambda\in\mathbb{C}$, and let $N_{\lambda}$ be a nilpotent matrix ($N_{\lambda}^{d}=0$). Then: 
\[
\left|\prod_{k=\lceil\left|\lambda\right|\rceil}^{m-1}\left(I+\frac{N_{\lambda}}{n_{0}+k\cdot r+\lambda}\right)\right|\leq C_{N_{\lambda}}\cdot\left(\ln\left(m\right)\right)^{d}
\]
for sufficiently large $m$, where $C_{N_{\lambda}}$ is a constant depending only on $N_{\lambda}$ and independent of $m$ (where $\left|\cdot\right|$ refers to the standard operator norm over complex matrices).
\end{lemma}
\begin{proof}

First if we let $a_{k}=\frac{1}{n_{0}+k\cdot r+\lambda}$, then:
\[
\left|\prod_{k=\lceil\left|\lambda\right|\rceil}^{m-1}\left(I+\frac{N_{\lambda}}{n_{0}+k\cdot r+\lambda}\right)\right|=\left|\prod_{k=\lceil\left|\lambda\right|\text{\ensuremath{\rceil}}}^{m-1}\left(I+a_{k}\cdot N_{\lambda}\right)\right|
\]

Observe that all matrices in the product commute and thus:
\begin{equation}\label{inq:someInq}
\left|\prod_{k=\lceil\left|\lambda\right|\text{\ensuremath{\rceil}}}^{m-1}\left(I+a_{k}\cdot N_{\lambda}\right)\right|\leq 
\left|I+\sum_{j=1}^{d-1}p_{j}\left(\vec{a}\right)\cdot\left(N_{\lambda}\right)^{j}\right| 
\end{equation}
where
\[
\forall j\geq 1: p_{j}\left(\vec{a}\right)=\sum_{\lceil\left|\lambda\right|\text{\ensuremath{\rceil}}\leq k_{1}<...<k_{j}\leq m-1}a_{k_{1}}\cdot...\cdot a_{k_{j}}
\]
It would thus suffice to provide the desired upper bound for the right hand side of \Cref{inq:someInq}, which we will focus on for the remainder of the proof.
By the triangle inequality and submultiplicativity of the operator norm:

\begin{equation}\label{eq:Upperbound2}
\left|I+\sum_{j=1}^{d-1}p_{j}\left(\vec{a}\right)\cdot\left(N_{\lambda}\right)^{j}\right|\leq 1+\sum_{j=1}^{d-1}\left|p_{j}\left(\vec{a}\right)\right|\cdot\left|N_{\lambda}\right|^{j}
\end{equation}
Using a standard bound (counting argument): 
\[
\left|p_{j}\left(\vec{a}\right)\right|=\left|\sum_{\lceil\left|\lambda\right|\text{\ensuremath{\rceil}}\leq k_{1}<...<k_{j}\leq m-1}a_{k_{1}}\cdot...\cdot a_{k_{j}}\right|\leq_{\triangle}\sum_{\lceil\left|\lambda\right|\text{\ensuremath{\rceil}}\leq k_{1}<...<k_{j}\leq m-1}\left|a_{k_{1}}\right|\cdot...\cdot\left|a_{k_{j}}\right|\leq\frac{1}{j!}\cdot\left(\sum_{k=\lceil\left|\lambda\right|\text{\ensuremath{\rceil}}}^{m-1}\left|a_{k}\right|\right)^{j}.
\]

Therefore, 
\begin{equation}
\label{eq:UpperBound3}
1+\sum_{j=1}^{d-1}\left|p_{j}\left(\vec{a}\right)\right|\cdot\left|N_{\lambda}\right|^{j}\leq1+\sum_{j=1}^{d-1}\frac{1}{j!}\cdot\left(\sum_{k=\lceil\left|\lambda\right|\text{\ensuremath{\rceil}}}^{m-1}\left|a_{k}\right|\right)^{j}\cdot\left|N_{\lambda}\right|^{j}\,
\end{equation}

Now to bound $\sum_{k=\lceil\left|\lambda\right|\text{\ensuremath{\rceil}}}^{m-1}\left|a_{k}\right|$: 
\[
\sum_{k=\lceil\left|\lambda\right|\text{\ensuremath{\rceil}}}^{m-1}\left|a_{k}\right|=\sum_{k=\lceil\left|\lambda\right|\text{\ensuremath{\rceil}}}^{m-1}\frac{1}{\left|n_{0}+k\cdot r+\lambda\right|}\leq\sum_{k=\lceil\left|\lambda\right|\text{\ensuremath{\rceil}}}^{m-1}\frac{1}{n_{0}+k\cdot r-\lceil\left|\lambda\right|\rceil}\leq\sum_{i=1}^{n_{0}+m\cdot r}\frac{1}{i}\leq\ln\left(n_{0}+m\cdot r\right)+2\leq2\cdot\ln\left(m\right).
\]
Therefore:
\[
1+\sum_{j=1}^{d-1}\frac{1}{j!}\cdot\left(\sum_{k=\lceil\left|\lambda\right|\text{\ensuremath{\rceil}}}^{m-1}\left|a_{k}\right|\right)^{j}\cdot\left|N_{\lambda}\right|^{j}\,\leq 1+\sum_{j=1}^{d-1}\frac{1}{j!}\cdot\left(2\ln\left(m\right)\right)^{j}\cdot\left|N_{\lambda}\right|^{j}
\]
This is at most $\left(\ln\left(m\right)\right)^{d}\cdot C_{N_{\lambda}}$ where $C_{N_{\lambda}}$ is a constant depending on $N_{\lambda}$ (specifically on its operator norm and $d$). Thus, chaining the inequalities from \Cref{eq:Upperbound2} and \Cref{eq:UpperBound3} completes the proof.
\end{proof}

Our convergence theorem uses the notion of a spectral projection. In short, a spectral projection is a projection onto a generalized eigenspace of a matrix. In our recurrences, as we've demonstrated in \Cref{subsec:SpectralCharacterization}, $r$ is a simple eigenvalue of the matrix. And thus, the generalized eigenspace corresponding to $r$ is just equal to the eigenspace corresponding to $r$. We use the fact that the spectral projection for eigenvalue $\lambda$ can be represented by taking the Jordan form of the matrix: $SJ S^{-1}$, replacing all Jordan blocks corresponding to $\lambda$ by the identity, and setting all other blocks to the zero matrix. For further reading, refer to \cite{0521386322}.
\begin{theorem}\label{thm:MatrixRecurrenceConvergance}
Assume a matrix recurrence: $\forall n\geq n_{0}:v_{n+r}=\left(I+\frac{1}{n}\cdot A\right)v_{n}$ where $r\in\sigma\left(A\right)$ is a simple positive eigenvalue, and all other eigenvalues $\lambda$ satisfy $\mathcal{R}e\left(\lambda\right)<r$. Then $\lim_{m\to\infty}\frac{v_{n_{0}+m\cdot r}}{n_{0}+m\cdot r}=P_{r}^{A}\left(\frac{v_{n_{0}}}{n_{0}}\right)$ where $P_{r}^{A}$ is the spectral projection operator onto $A$'s generalized eigenspace corresponding to eigenvalue $r$.
\end{theorem}
\begin{proof}
Observe that iterating the recurrence we get 
\[ v_{n_{0}+m\cdot r}=\left(\prod_{k=0}^{m-1}\left(I+\frac{1}{n_{0}+k\cdot r}\cdot A\right)\right)\cdot v_{n_{0}} \]
\[ \implies\frac{v_{n_{0}+m\cdot r}}{n_{0}+m\cdot r}=\left(\frac{1}{n_{0}+m\cdot r}\cdot\prod_{k=0}^{m-1}\left(I+\frac{1}{n_{0}+k\cdot r}\cdot A\right)\right)\cdot v_{n_{0}}.\]    
\end{proof}

We turn our focus to analyzing the matrix product $P_{m}=\frac{1}{n_{0}+m\cdot r}\cdot\prod_{k=0}^{m-1}\left(I+\frac{1}{n_{0}+k\cdot r}\cdot A\right)$.

Denote $\sigma\left(A\right)=\left\{ \lambda_{1}=r,\lambda_{2},...,\lambda_{\ell}\right\} $ and let $A=S\cdot J\cdot S^{-1}$ be the complex Jordan form for $A$ where $J$ is a block diagonal matrix with blocks $\left(J_{\lambda_{1}},...,J_{\lambda_{\ell}}\right)$ where $J_{\lambda_{i}}=\lambda_{i}\cdot I+N_{i}$ is the block corresponding to eigenvalue $\lambda_{i}$.

Then: 
\[
P_{m}=\frac{1}{n_{0}+m\cdot r}\cdot\prod_{k=0}^{m-1}\left(I+\frac{1}{n_{0}+k\cdot r}\cdot S\cdot J\cdot S^{-1}\right)=\frac{1}{n_{0}+m\cdot r}\cdot\prod_{k=0}^{m-1}\left(S\cdot\left(I+\frac{1}{n_{0}+k\cdot r}\cdot J\right)\cdot S^{-1}\right)
\]
\[
=S\cdot\left(\frac{1}{n_{0}+m\cdot r}\cdot\prod_{k=0}^{m-1}\left(I+\frac{1}{n_{0}+k\cdot r}\cdot J\right)\right)\cdot S^{-1}.
\]

Now, since $J$ is a block-diagonal matrix, so is $I+\frac{1}{n_{0}+k\cdot r}\cdot J$ for each $k\geq0$. This means that to analyze the middle product we can exponentiate each diagonal block separately.

First, for $J_{\lambda_{1}}=J_{r}=\left[r\right]$ (a $1\times1$ matrix since $r$ is a simple eigenvalue), we obtain: 
\[
\frac{1}{n_{0}+m\cdot r}\cdot\prod_{k=0}^{m-1}1+\frac{1}{n_{0}+k\cdot r}\cdot r=\frac{1}{n_{0}+m\cdot r}\cdot\frac{n_{0}+m\cdot r}{n_{0}}=\frac{1}{n_{0}}.
\]

Now, for $\lambda\neq\lambda_{1}$: Our goal is to show that $\frac{1}{n_{0}+m\cdot r}\cdot\prod_{k=0}^{m-1}\left(I+\frac{1}{n_{0}+k\cdot r}\cdot J\right)$ converges to the zero matrix as $m$ tends to infinity. For this it suffices to prove its operator norm approaches $0$.

First, let $\gamma$ denote $\left|\prod_{k=0}^{\lceil\left|\lambda\right|\text{\ensuremath{\rceil}}-1}\left(I+\frac{1}{n_{0}+k\cdot r}\cdot J\right)\right|$, and observe this is a quantity depending only on $n_{0},r,\lambda$ and $j$, i.e. it is independent of $m$.

Thus: 
\[
\left|\frac{1}{n_{0}+m\cdot r}\cdot\prod_{k=0}^{m-1}\left(I+\frac{1}{n_{0}+k\cdot r}\cdot J\right)\right|=\gamma\cdot\left|\frac{1}{n_{0}+m\cdot r}\cdot\prod_{k=\lceil\left|\lambda\right|\text{\ensuremath{\rceil}}}^{m-1}\left(I+\frac{1}{n_{0}+k\cdot r}\cdot J\right)\right|.
\]

It therefore suffices to show that $\left|\frac{1}{n_{0}+m\cdot r}\cdot\prod_{k=\lceil\left|\lambda\right|\text{\ensuremath{\rceil}}}^{m-1}\left(I+\frac{1}{n_{0}+k\cdot r}\cdot J\right)\right|$ converges to $0$ to conclude that the norm converges to $0$.

Note that for every $k\geq\lceil\left|\lambda\right|\text{\ensuremath{\rceil}}$: 
\[
I+\frac{1}{n_{0}+k\cdot r}\cdot J_{\lambda}=I+\frac{1}{n_{0}+k\cdot r}\cdot\left(\lambda\cdot I+N_{\lambda}\right)=I+\frac{\lambda}{n_{0}+k\cdot r}\cdot I+\frac{1}{n_{0}+k\cdot r}\cdot N_{\lambda}
\]
\[
=\left(1+\frac{\lambda}{n_{0}+k\cdot r}\right)\cdot\left(I+\frac{N_{\lambda}}{n_{0}+k\cdot r+\lambda}\right).
\]

Thus: 
\[
\left|\frac{1}{n_{0}+m\cdot r}\cdot\prod_{k=\lceil\left|\lambda\right|\text{\ensuremath{\rceil}}}^{m-1}\left(I+\frac{1}{n_{0}+k\cdot r}\cdot J\right)\right|=\left|\frac{1}{n_{0}+m\cdot r}\cdot\prod_{k=\lceil\left|\lambda\right|\text{\ensuremath{\rceil}}}^{m-1}\left(1+\frac{\lambda}{n_{0}+k\cdot r}\right)\cdot\left(I+\frac{N_{\lambda}}{n_{0}+k\cdot r+\lambda}\right)\right|
\]

\[
=\left|\frac{1}{n_{0}+m\cdot r}\cdot\prod_{k=\lceil\left|\lambda\right|\text{\ensuremath{\rceil}}}^{m-1}\left(1+\frac{\lambda}{n_{0}+k\cdot r}\right)\cdot\prod_{k=\lceil\left|\lambda\right|\text{\ensuremath{\rceil}}}^{m-1}\left(I+\frac{N_{\lambda}}{n_{0}+k\cdot r+\lambda}\right)\right|
\]

\[
=\frac{1}{n_{0}+m\cdot r}\cdot\underbrace{\mbox{\ensuremath{\left|\prod_{k=\lceil\left|\lambda\right|\text{\ensuremath{\rceil}}}^{m-1}\left(1+\frac{\lambda}{n_{0}+k\cdot r}\right)\right|}}}_{\mbox{A}}\cdot\underbrace{\mbox{\ensuremath{\left|\prod_{k=\lceil\left|\lambda\right|\text{\ensuremath{\rceil}}}^{m-1}\left(I+\frac{N_{\lambda}}{n_{0}+k\cdot r+\lambda}\right)\right|}}}_{\mbox{B}}.
\]

According to \Cref{lem:boundForConverganceI} and \Cref{lem:boundForConverganceII}, for sufficiently large $m$ we have: $A\leq m^{\frac{\mathcal{R}e\left(\lambda\right)}{r}}\cdot\alpha_{\lambda,n_{0},r}$, and $B\leq C_{N_{\lambda}}\cdot\left(\ln\left(m\right)\right)^{d}$ where $\alpha_{\lambda,n_{0},r}$ and $C_{N_{\lambda}}$ are independent of $m$.

Observe that 
\[
\lim_{m\to\infty}\frac{1}{n_{0}+m\cdot r}\cdot C_{N_{\lambda}}\cdot\left(\ln\left(m\right)\right)^{d}\cdot m^{\frac{\mathcal{R}e\left(\lambda\right)}{r}}\cdot\beta_{\lambda,n_{0},r}=0,
\]
and thus applying the squeeze rule we can conclude $\lim_{m\to\infty}\frac{1}{n_{0}+m\cdot r}\cdot A\cdot B=0$, as needed.

To summarize:
\[
\lim_{m\to\infty }\left(\frac{1}{n_{0}+m\cdot r}\cdot\prod_{k=0}^{m-1}\left(I+\frac{1}{n_{0}+k\cdot r}\cdot A\right)\right)\cdot v_{n_{0}}
\]
\[
= S\cdot diag\left(\frac{1}{n_{0}}\cdot I,\boldsymbol{0},...,\boldsymbol{0}\right)\cdot S^{-1}\cdot v_{n_{0}}=\frac{1}{n_{0}}\cdot P_{r}^{A}\cdot v_{n_{0}}=P_{r}^{A}\left(\frac{v_{n_{0}}}{n_{0}}\right).
\]

\subsection{Characterizing the average block size}\label{subsec:AverageBlockSize}

As mentioned in \Cref{subsec:UniformModeling}, inserting $(-\infty)$ into the structure prior to the workload affords having a neat recurrence. An additional assumption we make for the purpose of a cleaner analysis for even splitting is that the first batch has size exactly $\frac{B-1}{2}$.
Under this additional assumption $v_{n}$ only ever has support on $S=d+\left\langle r\right\rangle $. We are thus only interested in the behavior of $\frac{\left(v_{n}\right)_{S}}{n}$ as $n$ tends to infinity. We let $z_n$ denote $\left(v_{n}\right)_{S}$.

Observe that the recurrence for $z_{n}$ can be represented by:
\[
\forall n\geq d:z_{n+r}=\left(I+\frac{1}{n}\cdot A\left(B,r\right)_{S}\right)\cdot z_{n}
\]
 where $v_d=e_1$.

Combining \Cref{thm:MatrixRecurrenceConvergance} and \Cref{lem:FinalSpectralCharacterization}, we obtain 
\[
\lim_{m\to\infty}\frac{z_{n_{0}+m\cdot r}}{n_{0}+m\cdot r}=P_{r}^{A\left(B,r\right)_{S}}\left(\frac{z_{n_{0}}}{n_{0}}\right),
\]
where $P_{r}^{A\left(B,r\right)_{S}}(\cdot)$ is the spectral projection onto the generalized eigenspace of $A(B,r)$ corresponding to $r$ (for further reading on spectral projections, refer to \cite{0521386322}). 
According to \Cref{lem:FinalSpectralCharacterization}, $r$ is a simple eigenvalue of  $A\left(B,r\right)_{S}$, and thus the spectral projection is a projection onto the eigenspace corresponding to $r$.
Overall, we conclude that $\frac{z_{n_{0}+m\cdot r}}{n_{0}+m\cdot r}$ converges to some vector in the eigenspace corresponding to $r$. \Cref{lem:evensplitaveragebocksizeconverges} uses this fact to characterize the average block size using any eigenvector of $A(B,r)_S$.

\begin{lemma}
\label{lem:evensplitaveragebocksizeconverges}

Let $w=\left(d,d+1,...,B\right)$, and let $u$ be any eigenvector corresponding to eigenvalue $r$ for $A\left(B,r\right)_{S}$. 
Then the expected average block size converges to at least $\frac{\left\langle u,w_{S}\right\rangle }{\left\langle u,\vec{1}\right\rangle}$, where $\langle \cdot \rangle $ denotes the standard inner product.

\end{lemma}
\begin{proof}

Note that at any time step $n$ the identity $n=\sum_{k\in S}k\cdot x_{k}^{n}$ holds as an invariant.
Further, the average size of a block at time step $n$, is $\frac{n}{X^{n}}$
where $X^{n}=\sum_{k\in S}x_{k}^{n}$.
Observe: 
\[
\frac{\left\langle z_{n},w_{S}\right\rangle }{\left\langle z_{n},\vec{1}\right\rangle }=\frac{\sum_{k\in S}k\cdot\mathbb{E}\left(x_{k}^{n}\right)}{\sum_{k\in S}\mathbb{E}\left(x_{k}^{n}\right)}=\frac{\mathbb{E}\left(\sum_{k\in S}k\cdot x_{k}^{n}\right)}{\mathbb{E}\left(\sum_{k\in S}x_{k}^{n}\right)}=n\cdot\frac{1}{\mathbb{E}\left(\sum_{k\in S}x_{k}^{n}\right)}
\]

\[
\leq_{\text{Jensen}}n\cdot\mathbb{E}\left(\frac{1}{\sum_{k\in S}x_{k}^{n}}\right)=\mathbb{E}\left(\frac{n}{\sum_{k\in S}x_{k}^{n}}\right)=\mathbb{E}\left(\frac{n}{X^{n}}\right).
\]

Thus, the average blocks size converges to: 
\[
\lim_{m\to\infty}\frac{\left\langle z_{n_{0}+m\cdot r},w_{S}\right\rangle }{\left\langle z_{n_{0}+m\cdot r},\vec{1}\right\rangle }=\lim_{m\to\infty}\frac{\left\langle \frac{z_{n_{0}+m\cdot r}}{n_{0}+m\cdot r},w_{S}\right\rangle }{\left\langle \frac{z_{n_{0}+m\cdot r}}{n_{0}+m\cdot r},\vec{1}\right\rangle }.
\]

We know by the preceding discussion in \Cref{subsec:AverageBlockSize} that $\frac{z_{n_{0}+m\cdot r}}{n_{0}+m\cdot r}$ converges
to some vector in the eigenspace of $A\left(B,r\right)_{S}$ corresponding
to $r$. Denote this vector by $u^{\prime}$, and we thus have: 
\[
\lim_{m\to\infty}\frac{\left\langle \frac{z_{n_{0}+m\cdot r}}{n_{0}+m\cdot r},w_{S}\right\rangle }{\left\langle \frac{z_{n_{0}+m\cdot r}}{n_{0}+m\cdot r},\vec{1}\right\rangle }=\frac{\left\langle u^{\prime},w_{S}\right\rangle }{\left\langle u^{\prime},\vec{1}\right\rangle }
\]
Since $r$ is a simple eigenvalue, the eigenspace corresponding to
it is one-dimensional. This means there exists some scalar $c$ such
that $u^{\prime}=c\cdot u$, therefore the limit equals:
\[\frac{\left\langle c\cdot u,w_{S}\right\rangle }{\left\langle c\cdot u,\vec{1}\right\rangle }=\frac{\left\langle u,w_{S}\right\rangle }{\left\langle u,\vec{1}\right\rangle }. \]
\end{proof}

Our goal is thus now to characterize an arbitrary eigenvector in the eigenspace corresponding to $r$ to in order to provide a lower bound for: 
\[
\frac{\left\langle u,w_{S}\right\rangle }{\left\langle u,\vec{1}\right\rangle }=\frac{\sum_{k\in S}k\cdot u_{k}}{\sum_{k\in S}u_{k}}.
\]

\subsection{Lower bounding the average block size}\label{subsec:LowerBounds}\label{subsec:LowerBounds}

We've established the average block size converges to $\frac{\sum_{k\in S}k\cdot u_{k}}{\sum_{k\in S}u_{k}}$ for any eigenvector $u$ of $A\left(B,r\right)_{S}$.

Let $u$ be a positive eigenvector of $A\left(B,r\right)_{S}$.

\begin{lemma}
\label{lem:step-r-recurrence}

For every $k\in S$ such that $k\geq d+r$, we have $u_{k}=\frac{k-r}{k+r}\cdot u_{k-r}$.

\end{lemma}
\begin{proof}

Let $k\in S$ such that $k\geq d+r$, let $\hat{u}$ be the expansion
of $u$ into a $d$-dimensional vector such that $\forall k\in S:\hat{u}_{k}=u_{k}$,
and in all other position $\hat{u}$ is $0$. Observe that $A_{k,*}\cdot\hat{u}=\left(A\left(B,r\right)_{S}\right)_{k,*}\cdot u$
where $A_{k,*}$ denotes the $k^{th}$ row of $A$. 

Observe that for every such $k$, $A_{k,k}=-k,A_{k,k-r}=k-r$ and
all other entries are $0$.

Therefore, $\left(A\left(B,r\right)_{S}\right)_{k,*}\cdot u=A_{k,*}\cdot\hat{u}=-k\cdot u_{k}+\left(k-r\right)u_{k-r}$.

Since $u$ is an eigenvector of $A\left(B,r\right)_{S}$ we know that
$\left(A\left(B,r\right)_{S}\right)_{k,*}\cdot u=r\cdot u_{k}$,
so overall: $r\cdot u_{k}=-k\cdot u_{k}+\left(k-r\right)u_{k-r}\implies u_{k}=\frac{k-r}{k+r}\cdot u_{k-r}$
and this holds for every $k\in S$ such that $k\geq d+r$.
\end{proof}

Now observe that it suffices to lower bound the term of interest, i.e. $\frac{\sum_{k\in S}k\cdot u_{k}}{\sum_{k\in S}u_{k}}$, when restricting the sums to some class in a partition of $S$, as formulated by the following lemma:
\begin{lemma}\label{lem:AveragePartitioning}

Let $C_{1},...,C_{\ell}$ be some partition of $S$. Then 
\[
\frac{\sum_{k\in S}k\cdot u_{k}}{\sum_{k\in S}u_{k}}\geq min\left\{ \frac{\sum_{k\in C_{i}}k\cdot u_{k}}{\sum_{k\in C_{i}}u_{k}}:i=1,...,\ell\right\}. 
\]

\end{lemma}
\begin{proof}

Note: 
\[
\frac{\sum_{k\in S}k\cdot u_{k}}{\sum_{k\in S}u_{k}}=\sum_{C_{j}}\lambda_{j}\cdot\frac{\sum_{k\in C_{j}}k\cdot u_{k}}{\sum_{k\in C_{j}}u_{k}}
\]
 where $\lambda_{j}=\frac{\sum_{k\in C_{j}}u_{k}}{\sum_{k\in S}u_{k}}$.
Now observe that $\sum_{C_{j}}\lambda_{j}=1$, and thus $\frac{\sum_{k\in S}k\cdot u_{k}}{\sum_{k\in S}u_{k}}$
is a weighted average of $\left\{ \frac{\sum_{k\in C_{1}}k\cdot u_{k}}{\sum_{k\in C_{1}}u_{k}},...,\frac{\sum_{k\in C_{\ell}}k\cdot u_{k}}{\sum_{k\in C_{\ell}}u_{k}}\right\} $
so is therefore at least the minimum term.
\end{proof}

In the partition of $S$ we will use, each class $C$ will be of the form $\left\{ j_{0},j_{1},...,j_{L}\right\} $ where $d\leq j_{0}\leq d+r-1$, $\forall i:j_{i}=j_{0}+r\cdot i$, and $L$ is the maximum integer for which $j_{0}+L\cdot i\leq B$.

According to \Cref{lem:AveragePartitioning}, it suffices to set such a class $C$ and provide a lower bound for $\frac{\sum_{i=0}^{L}j_{i}\cdot u_{j_{i}}}{\sum_{i=0}^{L}u_{j_{i}}}$.

By \Cref{lem:step-r-recurrence},  $u_{j_{i}}=\frac{j_{i-1}}{j_{i+1}}\cdot u_{j_{i-1}}$ for every $1\leq i\leq L$, letting  $j_{L+1}$ denote $j_{L}+r$.

Thus, for every $\ell\geq2:$ 
\[
u_{j_{\ell}}=u_{j_{0}}\cdot\prod_{i=1}^{\ell}\frac{j_{i-1}}{j_{i+1}}=u_{j_{0}}\cdot\frac{j_{0}\cdot j_{1}}{j_{\ell}\cdot j_{\ell+1}}.
\]

Observe that the above equality also holds for $\ell=0,1$.

Overall we get: for every $\ell\geq0:u_{j_{\ell}}=\frac{j_{0}\cdot\left(j_{0}+r\right)}{j_{\ell}\cdot\left(j_{\ell}+r\right)}\cdot u_{j_{0}}$.

So: 
\[
\frac{\sum_{i=0}^{L}j_{i}\cdot u_{j_{i}}}{\sum_{i=0}^{L}u_{j_{i}}}=\frac{\sum_{i=0}^{L}j_{i}\cdot\frac{j_{0}\cdot\left(j_{0}+r\right)}{j_{i}\cdot\left(j_{i}+r\right)}\cdot u_{j_{0}}}{\sum_{i=0}^{L}\frac{j_{0}\cdot\left(j_{0}+r\right)}{j_{i}\cdot\left(j_{i}+r\right)}\cdot u_{j_{0}}}
\]

\begin{equation}\label{eq:TermToBoundForAverageFill}
=\frac{j_{0}\cdot\left(j_{0}+r\right)\cdot u_{j_{0}}\cdot\sum_{i=0}^{L}\frac{1}{j_{i}+r}}{j_{0}\cdot\left(j_{0}+r\right)\cdot u_{j_{0}}\cdot\sum_{i=0}^{L}\frac{1}{j_{i}\cdot\left(j_{i}+r\right)}}=\frac{\sum_{i=0}^{L}\frac{1}{j_{i}+r}}{\sum_{i=0}^{L}\frac{1}{j_{i}\cdot\left(j_{i}+r\right)}}
\end{equation}

We now provide three lower bounds for the average block size, that reflect different parameter regimes for $r$. The first lower bound is favorable to the others when $r<0.0058B$. The second lower bound is favorable when $r<0.21B$, and the third lower bound is favorable in the rest of the regime we analyze for even splitting, i.e. $0.21B\leq r<0.5B$.
The first two lower bounds are obtained, by lower bounding the right hand side of \Cref{eq:TermToBoundForAverageFill}. Then, we consider a refinement to the partition we described for $S$ in order to obtain a third lower bound.

\subsubsection{First lower bound - small values of $r$}\label{subsubsec:SmallLowerBound}

Our goal in this part is to prove that for every $r\leq0.0058B$ we have:
\[
\frac{\sum_{i=0}^{L}\frac{1}{j_{i}+r}}{\sum_{i=0}^{L}\frac{1}{j_{i}\cdot\left(j_{i}+r\right)}}\geq B\left(\ln(2)-\frac{5r}{B}\right)
\]
Begin by observing:
\[
\sum_{i=0}^{L}\frac{1}{j_{i}\left(j_{i}+r\right)}=\frac{1}{r}\cdot\left(\frac{1}{j_{0}}-\frac{1}{j_{L+1}}\right).
\]
Thus: 
\[
\frac{\sum_{i=0}^{L}\frac{1}{j_{i}+r}}{\sum_{i=0}^{L}\frac{1}{j_{i}\cdot\left(j_{i}+r\right)}}=\frac{r\cdot\sum_{i=0}^{L}\frac{1}{j_{i}+r}}{\frac{1}{j_{0}}-\frac{1}{j_{L+1}}}\geq\frac{r\cdot\sum_{i=0}^{L}\frac{1}{j_{i}+r}}{\frac{1}{d}-\frac{1}{B+r}}.
\]
Now note: 
\[
\sum_{i=0}^{L}\frac{1}{j_{i}+r}=\sum_{i=1}^{L+1}\frac{1}{j_{0}+i\cdot r}\geq\int_{j_{0}+r}^{j_{0}+\left(L+1\right)\cdot r}\frac{dx}{x}\geq\ln\left(B+1\right)-\ln\left(j_{0}+r\right)\geq\frac{1}{r}\cdot\ln\left(\frac{B+1}{d+2r}\right).
\]
So: 
\[
\frac{\sum_{i=0}^{L}\frac{1}{j_{i}+r}}{\sum_{i=0}^{L}\frac{1}{j_{i}\cdot\left(j_{i}+r\right)}}\geq\frac{r\cdot\sum_{i=0}^{L}\frac{1}{j_{i}+r}}{\frac{1}{d}-\frac{1}{B+r}}\geq\frac{\ln\left(\frac{B+1}{d+2r}\right)}{\frac{1}{d}-\frac{1}{B+r}}\geq\frac{\ln\left(\frac{2\left(B+1\right)}{B+1+4r}\right)}{\frac{2}{B+1}-\frac{1}{B+r}}=\ln\left(\frac{2}{1+\frac{4r}{B+1}}\right)\cdot\frac{\left(B+1\right)\cdot\left(B+r\right)}{B+2r-1}.
\]
Observe that $\frac{\left(B+1\right)\cdot\left(B+r\right)}{B+2r-1}\geq B\cdot\left(1-\frac{r}{B}\right)$, and $\ln\left(\frac{2}{1+\frac{4r}{B+1}}\right)\geq\ln\left(2\right)-\frac{4r}{B+1}$ (as $\ln\left(\frac{2}{1+x}\right)\geq\ln\left(2\right)-x$ for every $x\geq0$).
So:
\[
\frac{\sum_{i=0}^{L}\frac{1}{j_{i}+r}}{\sum_{i=0}^{L}\frac{1}{j_{i}\cdot\left(j_{i}+r\right)}}\geq B\cdot\left(\ln\left(2\right)-\frac{4r}{B+1}\right)\cdot\left(1-\frac{r}{B}\right)\geq B\cdot\left(\ln\left(2\right)-\frac{r\cdot\left(4+\ln\left(2\right)\right)}{B}\right)\geq B\cdot\left(\ln\left(2\right)-\frac{5r}{B}\right).
\]

\subsubsection{Second lower bound - moderate values of $r$}\label{subsubsec:ModerateLowerBound}

Our goal in this part is to prove that for every $r\leq 0.21B$ we have:
\[
\frac{\sum_{i=0}^{L}\frac{1}{j_{i}+r}}{\sum_{i=0}^{L}\frac{1}{j_{i}\cdot\left(j_{i}+r\right)}}\geq \frac{2B\cdot\left(B+1\right)}{3B+1+2r}
\]
Begin by observing:
\[
\sum_{i=0}^{L}\frac{1}{j_{i}\left(j_{i}+r\right)}=\frac{1}{r}\cdot\left(\frac{1}{j_{0}}-\frac{1}{j_{L+1}}\right)=\frac{j_{L+1}-j_{0}}{r\cdot j_{0}\cdot j_{L+1}}
\]
Thus: 
\[
\frac{\sum_{i=0}^{L}\frac{1}{j_{i}+r}}{\sum_{i=0}^{L}\frac{1}{j_{i}\cdot\left(j_{i}+r\right)}}=r\cdot\frac{j_{0}\cdot j_{L+1}}{j_{L+1}-j_{0}}\cdot\sum_{i=0}^{L}\frac{1}{j_{i}+r}
\]

\[
=r\cdot\frac{j_{0}\cdot j_{L+1}}{j_{0}+r\cdot\left(L+1\right)-j_{0}}\cdot\sum_{i=0}^{L}\frac{1}{j_{i}+r}=j_{0}\cdot j_{L+1}\cdot\left(\frac{1}{L+1}\cdot\sum_{i=1}^{L+1}\frac{1}{j_{i}}\right).
\]
According to Jensen's inequality the average of reciprocals is at least the reciprocal of the average, and observe: $\frac{1}{L+1}\cdot\sum_{i=1}^{L+1}j_{i}=\frac{j_{1}+j_{L+1}}{2}$.
Thus: 
\[
j_{0}\cdot j_{L+1}\cdot\left(\frac{1}{L+1}\cdot\sum_{i=1}^{L+1}\frac{1}{j_{i}}\right)\geq j_{0}\cdot j_{L+1}\cdot\frac{2}{j_{1}+j_{L+1}}=\frac{2\cdot j_{0}\cdot j_{L+1}}{j_{0}+j_{L+1}+r}.
\]
Observe that the 2D function $f\left(x,y\right)=\frac{2\cdot x\cdot y}{x+y+r}$ is increasing in both $x,y$ for $x>0,y>0$ since both partial derivatives are positive. Thus, if we observe the function in an axis-aligned rectangle contained in the first quadrant with bottom left corner $\left(x_{0},y_{0}\right)$, the function is minimized in $\left(x_{0},y_{0}\right)$. Observe the function in the axis-aligned rectangle given by $d\leq x\leq d+r-1$ and $B\leq y\leq B+2r$; then $f\left(x,y\right)$ is minimized at $\frac{2\cdot d\cdot B}{d+B+r}=\frac{2\left(B+1\right)\cdot B}{3B+1+2r}$. Note that $\left(j_{0},j_{L+1}\right)$ is within this rectangle and therefore: 
\[
\frac{2\cdot j_{0}j_{L+1}}{j_{0}+j_{L+1}+r}\geq\frac{2\left(B+1\right)\cdot B}{3B+1+2r}.
\]

\subsubsection{Third lower bound - larger values of $r$} \label{subsubsec:LargeLowerBound}
Our goal in this section is to show that for every $1\leq r\leq \frac{5}{12}B$ (see \footnote{We note that as $r$ grows, some of the classes we define may become singletons. But this only occurs if $j_0+r\geq B\iff j_0\geq B-r$. Under the assumption that $r\leq \frac{5}{12}B$ this would imply $j_0\geq \frac{7}{12}B$. So, the contribution of such a class for the weighted average without any further refinement is $\frac{\sum_{i=0}^{L}\frac{1}{j_{i}+r}}{\sum_{i=0}^{L}\frac{1}{j_{i}\cdot\left(j_{i}+r\right)}}=j_0\geq \frac{7}{12}B$, respecting the stated bound. The refinement is used to argue about the remaining classes.}) the average block size converges to a value that is at least $\frac{7}{12}B$. 
Towards this end, we consider a refinement of the partition on $S$ we considered in \Cref{subsec:AverageBlockSize}. Recall each class $C$ is of the form $\{j_0,j_1,...,j_\ell \}$ where $d\leq j_0\leq d+r-1$, and for every $i:j_i=j_0+r\cdot i$. Further, recall that 
\[
\forall \ell\geq 0:u_{j_\ell}=\frac{j_0\cdot (j_0+r)}{j_\ell\cdot (j_\ell +r)}\cdot u_{j_0}:=\frac{1}{j_\ell \cdot (j_\ell +r)}\cdot \kappa _C.\hspace{0.3cm} (\star)
\]
Our refinement is obtained by further partitioning each such class $C$ into \[
\{\{j_i,j_{L-i}\}:i\in\{0,1,...,\lfloor L/2\rfloor \}\}
\]
According to \Cref{lem:AveragePartitioning}, it suffices to provide a lower bound on
\[
F(a,b,r)=\frac{a\cdot u_a+b\cdot u_b}{u_a+u_b}
\]
where $a=j_i$ and $b=j_{L-i}$ for some $i\in \{0,...,\lfloor L/2\rfloor\}$.
Plugging in $(\star)$ we obtain:
\begin{align*}
F(a,b,r)
&= \frac{a \cdot \frac{\kappa_C}{a(a+r)} + b \cdot \frac{\kappa_C}{b(b+r)}}
          {\frac{\kappa_C}{a(a+r)} + \frac{\kappa_C}{b(b+r)}}, 
\end{align*}
\[
=\frac{\frac{1}{a+r} + \frac{1}{b+r}}{\frac{1}{a(a+r)} + \frac{1}{b(b+r)}}=\frac{ab(a+b+2r)}{a^2 + b^2 + r(a+b)}.
\]
\begin{lemma}
$F(a,b,r)$ is minimized at $a=j_0$ and $b=j_L$.
\end{lemma}
\begin{proof}
Fixing $r>0$ and $s>0$, we can show that 
\[
F(a,s-a,r) = \frac{a(s-a)(s+2r)}{a^2 + (s-a)^2 + rs}.
\]
is strictly increasing in $a$ for $a\in(0,s/2)$ by taking a partial derivative on $a$.
So, fixing $r$ and $s=j_0+j_L$, we have for every $i\in \{0,...,\lfloor L/2\rfloor \}$:
\[
F(j_i,(j_0+j_L)-j_i,r)\geq F(j_0,j_L,r).
\]
Observe that $(j_0+j_L)-j_i=j_{L-i}$ and thus for every $i\in \{0,...,\lfloor L/2\rfloor \}$:
\[
F(j_i,j_{L-i},r)\geq F(j_0,j_L,r).
\]
\end{proof}
 To proceed, we set:
\[
x := \frac{j_0}{B},\quad y := \frac{j_L}{B},\quad \alpha := \frac{r}{B},
\]
Then, we define:
\[
f(x,y,\alpha) := \frac{F(j_0,j_L,r)}{B}
= \frac{x y (x+y+2\alpha)}{x^2 + y^2 + \alpha(x+y)}.
\]
What's left to prove now is that $f(x,y,\alpha)\geq\frac{7}{12}$ in the feasible domain, i.e:
\[
\frac{1}{2} \le x \le 1,\quad 1- \alpha < y < 1,\quad y-x \geq \alpha,\quad 0 < \alpha < \frac{1}{2},
\]
This would imply the average block size is at least $\frac{7}{12}\cdot B$, as needed.
\begin{lemma}\label{lem:f-min}
Let
\[
f(x,y,\alpha) := \frac{x y (x+y+2\alpha)}{x^2 + y^2 + \alpha(x+y)}.
\]
Assume
\[
\frac{1}{2} \le x \le 1,\quad 0 < \alpha < \frac{1}{2},\quad
1-\alpha \leq y \leq 1, \quad y - x \ge \alpha.
\]
Then
\[
f(x,y,\alpha) \;\ge\; \frac{7}{12}.
\]
\end{lemma}

\begin{proof}
We begin by showing $f$ is monotonic in $y$ in the stated regime.
\[
\frac{\partial f}{\partial y}
= \frac{x(\alpha + x)\bigl(2\alpha x + x^2 + 2xy - y^2\bigr)}
{\bigl(x^2 + y^2 + \alpha(x+y)\bigr)^2}.
\]
The sign of the derivative is thus determined by $2\alpha x+x^2+2xy-y^2$.
Observe:
\[
2\alpha x+x^2+2xy-y^2=2\alpha x+x^2\left(1+2\cdot \frac{y}{x}-\left(\frac{y}{x}\right)^2\right).
\]
Now, since $\frac{1}{2}\leq x<y<1$, we have $\frac{y}{x}\in[1,2]$. Observe that $1+2t-t^2\geq1$ on $[1,2]$. Thus:
\[
2\alpha x+x^2\left(1+2\cdot \frac{y}{x}-\left(\frac{y}{x}\right)^2\right)\geq 2\alpha x+x^2>0.
\]
So, $f$ is strictly increasing in $y$ in the stated regime.
Now, given our constraints, i.e. $y\geq 1-\alpha$,$y\geq x+\alpha,y\geq 1-\alpha$, and thus $y$ is minimized at $y_{\min}:=\max\{x+\alpha,1-\alpha\}$. 
Therefore $f$ is minimized at $f(x,y_{\min},\alpha)$, and it is sufficient to show this value is at least $7/12$. We divide into two cases: 
\paragraph{Case I. $x\geq 1-2\alpha$:}
In this case $y_{\min}=x+\alpha$.
Observe that since $x\leq y-\alpha$ and $y<1$, we know that $x<1-\alpha$.
Overall we have:
\[
f(x,y_{\min},\alpha) = f(x,x+\alpha,\alpha)
= \frac{x(3\alpha+2x)}{2(x+\alpha)} =: g(x,\alpha),
\]
with
\[
x \in \bigl[\max\{1/2,1-2\alpha\},1-\alpha\bigr],\qquad 0<\alpha<1/2.
\]
Now, if $\alpha>\frac{1}{4}$, then $x\in \left[\frac{1}{2},1-\alpha\right]$. Computing partial derivatives reveals that $g$ is increasing in both $x$ and $\alpha$, as long as $x,\alpha>0$, and thus $g(x,\alpha)\geq g\left(\frac{1}{2},\frac{1}{4}\right)=\frac{7}{12}$.
Otherwise, $\alpha \leq \frac{1}{4}$, and then $x\in \left[1-2\alpha,1-\alpha\right]$, and since 
$g$ is increasing in $x$:
$g(x,\alpha)\geq g(1-2\alpha ,\alpha)=-\frac{(\alpha-2)(2\alpha-1)}{2\alpha-2}$. Observe $-\frac{(\alpha-2)(2\alpha-1)}{2\alpha-2}\geq\frac{7}{12}$ for every $\alpha\leq\frac{1}{4}$ (equality occurs at $\alpha=
\frac{1}{4}$).
\paragraph{Case II. $x< 1-2\alpha$:}
In this case $y_{\min}=1-\alpha$, and thus:
\[
f(x,y_{\min},\alpha)=f(x,1-\alpha,\alpha):=h(x,\alpha)=\frac{x(1-\alpha)(x+1+\alpha)}{x^2 + 1 - \alpha + \alpha x}.
\]
Taking a partial derivative on $h$ reveals it is increasing in $x$ for $x\in\left[\frac{1}{2},1-2\alpha\right]$, and thus $h(x,\alpha)\geq h(\frac{1}{2},\alpha):=\varphi(\alpha)=\frac{2\alpha^2+\alpha-3}{2\alpha-5}$. 
Now, one can easily verify that $\varphi(\alpha)\geq\frac{7}{12}$ (equality occurs at $\alpha=\frac{1}{4}$).
\end{proof}
\section{Algorithms for larger batch sizes}\label{sec:LargeBatchSizes}
Our analysis for even splitting shows that as long as $r\leq\left(\frac{1}{2}-\varepsilon\right)\cdot B$, for some constant $\varepsilon>0$, even splitting will produce fullness of at least $\frac{1}{2}+\delta$ for a constant $\delta>0$. We now proceed to present algorithms that can beat the trivial $0.5$ fullness bound for larger values of $r$, presenting different algorithms for different regimes. These algorithms are $r$ aware, in the sense that $r$ is used explicitly in the algorithm's description. Note that large batch sizes admit simple algorithms, which are far easier to analyze. Refer to \Cref{prelims:LargeBatches} for a high-level overview.

\subsection{Uneven splitting for moderately large batches} \label{subsec:UnevenSplit}
 We first introduce a subroutine that is a generic description of how our uneven splitting algorithms handle splits that occur during batches of insertions. The subroutine receives as input a memory block $A$ containing $B$ keys in sorted order, i.e. a full memory block, a key $k$ to be inserted into $A$, and targets $f_L$,$f_R$ such that $f_L+f_R>B$. The goal of the subroutine, is to split the block, insert $k$, and then handle the following insertions in the batch such that the two memory blocks created as result of the split, $A_L$ and $A_R$, contain $f_L$ and $f_R$ keys. See \Cref{fig:targetsplit}.

\begin{algorithm}
\caption{\textsc{TargetSplit}($\left(A=k_1<...<k_B\right),k,f_L,f_R$)}
\label{sub:TargetSplit}
\begin{algorithmic}[1]
\State Let $j=\left|\{k_i:i\in[B],k_i<k\}\right|$
\If{$j\geq f_L$}
\State Split $A$ into $A_L=k_1,...,k_{f_L}$ and $A_R=k_{f_{L}+1},....,k_{B}$.
\State Insert $k$ and the following $f_L+f_R-B-1$ keys in the batch into $A_R$.
\ElsIf{$j\leq B-f_R$}
\State Split $A$ into $A_L=k_{1},...,k_{B-f_{R}}$ and $A_R=k_{B-f_{R}+1},....,k_{B}$.
\State Insert $k$ and the following $f_L+f_R-B-1$ keys in the batch into $A_L$.
\Else
\State Split $A$ into $A_L=k_{1},...,k_{j}$ and $A_R=k_{j+1},...,k_{B}$.
\State Insert $k$ and the following $f_L-j-1$ keys in the batch into $A_L$.
\State After that, insert the next $f_R+j-B$ keys in the batch into $A_R$.
\EndIf

\end{algorithmic}
\end{algorithm}
\begin{lemma}\label{lem:TargetSplit}
Assume we call \textsc{TargetSplit} during an insertion batch, and at least $f_L+f_R-B-1$ additional insertions remain in the batch (i.e the next  $f_L+f_R-B-1$ insertions will take place between $k$ and $k_{j+1}$ in ascending order). Then, following the call to \textsc{TargetSplit}, $A_L$ contains $f_L$ keys and $A_R$ contains $f_R$ keys.
\end{lemma}
\begin{proof}
If the condition at (2) is satisfied, then we make the split such that $A_L$ contains $f_L$ elements. Further, since $j\geq f_L$, $k_{f_L+1}\leq k_{j+1}$, i.e $k_{j+1}$ is in $A_R$. So we can insert $k$ and the next keys in the batch into $A_R$. Following the insertions of $k$ and the following $B-f_L-f_R-1$ keys, $A_L$ contains $f_L$ keys, and $A_R$ contains $B-f_L+B+f_L+f_R-B=f_R$ keys.
We can similarly show that this also holds if the conditions in (5) or (8) are met instead.
\end{proof}

\begin{figure*}[h!]
    \centering
    \resizebox{\textwidth}{!}{
        \begin{tikzpicture}[
    cell/.style={
        draw, 
        minimum width=0.75cm, 
        minimum height=0.8cm, 
        inner sep=0pt, 
        anchor=west,
        font=\small
    },
    dot/.style={
        circle, 
        draw=black, 
        inner sep=0pt, 
        minimum size=4pt
    },
    odot/.style={dot, fill=orange!80!white},
    bdot/.style={dot, fill=cyan!70!black},
    label text/.style={
        font=\bfseries\boldmath\scriptsize, 
        anchor=north,       
        yshift=-0.1cm,      
        text height=1.5ex,  
        text depth=0.5ex    
    }
]

    \coordinate (start) at (-4.125, 0); 

    \node[cell] (t1) at (start) {}; \node[odot] at (t1) {}; \node[label text] at (t1.south) {1};
    \node[cell] (t2) [right=0pt of t1] {$\dots$};
    \node[cell] (t3) [right=0pt of t2] {}; \node[odot] at (t3) {}; \node[label text] at (t3.south) {$j$};
    \node[cell] (t4) [right=0pt of t3] {}; \node[odot] at (t4) {}; \node[label text] at (t4.south) {$j{+}1$};
    \node[cell] (t5) [right=0pt of t4] {$\dots$};
    \node[cell] (t6) [right=0pt of t5] {}; \node[odot] at (t6) {}; 
    \node[cell] (t7) [right=0pt of t6] {$\dots$};
    \node[cell] (t8) [right=0pt of t7] {$\dots$};
    \node[cell] (t9) [right=0pt of t8] {}; \node[odot] at (t9) {}; \node[label text] at (t9.south) {$B$};

    \coordinate (split_point) at (t3.north east);
    \node (k_label) [above=0.8cm of split_point, xshift=-0.8cm, font=\Large\bfseries\boldmath] {$k$};
    \draw[->, thick, >=stealth, color=teal!80!black] (k_label.east) to[out=0, in=90] (split_point);

    \node (dots) [below=0.4cm of t5, font=\Large\bfseries] {$\vdots$};
    
    \coordinate (bl_start) at (-7.5, -3.5);
    \node[cell] (l1) at (bl_start) {}; \node[odot] at (l1) {}; \node[label text] at (l1.south) {1};
    \node[cell] (l2) [right=0pt of l1] {$\dots$};
    \node[cell] (l3) [right=0pt of l2] {}; \node[odot] at (l3) {}; \node[label text] at (l3.south) {$j$};
    \node[cell] (l4) [right=0pt of l3] {}; \node[bdot] at (l4) {}; \node[label text] at (l4.south) {$j{+}1$};
    \node[cell] (l5) [right=0pt of l4] {$\dots$};
    \node[cell] (l6) [right=0pt of l5] {}; \node[bdot] at (l6) {}; \node[label text] at (l6.south) {$f_L$};
    \node[cell] (l7) [right=0pt of l6] {}; 
    \node[cell] (l8) [right=0pt of l7] {}; 

    \coordinate (br_start) at (0.0, -3.5);
    \node[cell] (r1) at (br_start) {}; \node[bdot] at (r1) {}; \node[label text] at (r1.south) {1};
    \node[cell] (r2) [right=0pt of r1] {$\dots$};
    \node[cell] (r3) [right=0pt of r2] {}; \node[bdot] at (r3) {};
    \node[cell] (r4) [right=0pt of r3] {}; \node[odot] at (r4) {}; \node[label text] at (r4.south) {$(f_R{+}j{-}B)$};
    \node[cell] (r5) [right=0pt of r4] {$\dots$};
    \node[cell] (r6) [right=0pt of r5] {$\dots$};
    \node[cell] (r7) [right=0pt of r6] {}; \node[odot] at (r7) {}; \node[label text] at (r7.south) {$f_R$};
    \node[cell] (r8) [right=0pt of r7] {}; 

    \draw[->, thick, >=stealth] (dots.south) -- ++(225:1.55cm);
    \draw[->, thick, >=stealth] (dots.south) -- ++(315:1.55cm);

\end{tikzpicture}
    }
    \caption{TargetSplit}
    \label{fig:targetsplit}
\end{figure*}
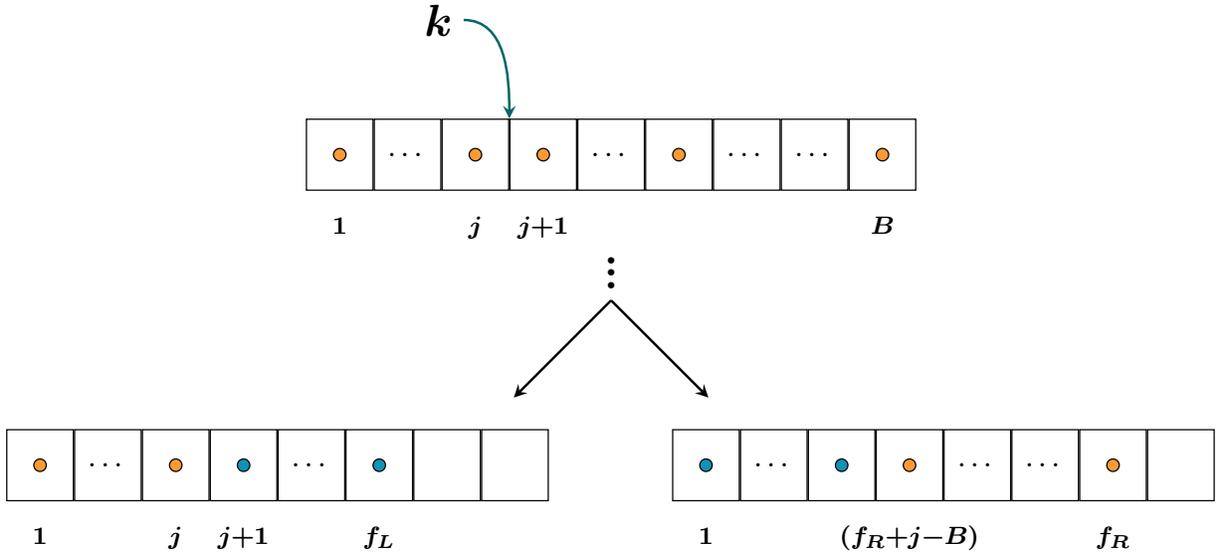

 We now turn to present our algorithms utilizing \textsc{TargetSplit} in \Cref{subsec:FirstLargeRegime} and \Cref{subsec:SecondLargeRegime}. Note that for the purpose of neat recurrences, we assume a dummy $(-\infty)$ insertion prior to the workload, similar to what we described in \cref{subsec:UniformModeling}. Further, for all insertions occurring outside of \textsc{TargetSplit}, these algorithms break ties to the left block.

\subsubsection{Regime I. $\frac{B}{3}<r\leq\frac{B}{2}$} \label{subsec:FirstLargeRegime}

For this regime, we maintain the invariant that immediately after each batch of insertions has completed, every block has either size $r$ or $2r$. 
Following the first batch we have a single block of size $r$. If a batch begins in a block of size $r$, at the end of it the block becomes of size $2r$, as it does not overflow since $r\leq B/2$.
If we insert into a block of size $2r$, we call $\textsc{TargetSplit}(A,k,r,2r)$ when the block becomes full. At this point there are $3r-B-1$ insertions remaining in the batch and we can thus apply \Cref{lem:TargetSplit} to conclude the result is two blocks of size $r$,$2r$, maintaining the invariant.

We can use the tools we saw for the analysis of even splitting to analyze the behavior of this block structure.
Specifically, let $v_{n}=\left(\mathbb{E}\left(x_{1}^{n}\right),\mathbb{E}\left(x_{2}^{n}\right)\right)$, where $x_{1}^{n},x_{2}^{n}$ denote the number of blocks of size $r,2r$ after $n$ insertions have occurred.

Let $x_{1}^{n},x_{2}^{n}$ denote the number of blocks of size $r,2r$ after $n$ insertions have occurred.
Then, similar to \Cref{subsec:UniformModeling}, we can obtain neat scalar recurrences for each component of $v_n$, and thus obtain following matrix recurrence: $v_{n+r}=\left(I+\frac{1}{n}\cdot A\right)\cdot v_{n}$ where $A=\begin{Bmatrix}-r & 2r\\
r & 0
\end{Bmatrix}$.

$A$ has two eigenvalues: $\lambda_{1}=r,\lambda_{2}=-2r$, and we can now use \Cref{thm:MatrixRecurrenceConvergance} and similarly to \Cref{subsec:AverageBlockSize} we can conclude the average block size converges to $\frac{\left\langle \vec{u},w\right\rangle }{\left\langle \vec{u},\vec{1}\right\rangle }$, where $w=\left(r,2r\right)$ and $\vec{u}$ is any eigenvector corresponding to $\lambda_{1}=r$.

We can choose $\vec{u}=\left(1,1\right)$, and get average block size $\frac{\left\langle \left(1,1\right),\left(r,2r\right)\right\rangle }{\left\langle \left(1,1\right),\vec{1}\right\rangle }=\frac{r+2r}{2}=\frac{3r}{2}$.

\subsubsection{Regime II. $\frac{2}{5}\cdot B<r\protect\leq\frac{2}{3}\cdot B$} \label{subsec:SecondLargeRegime}

In this regime we are able to enforce that all blocks have sizes in $\left\{ \frac{r}{2},r,\frac{3r}{2}\right\}$ by employing the following splitting rules. Every time a split needs to occur in a memory block $A$ during an insertion batch, we check the number of keys $A$ had before the batch started. Note that it could not be $\frac{r}{2}$ otherwise a split would not have occurred during the batch. So, assuming the invariant holds, it can only be $r$ or $2r$. If $A$ had $r$ keys before, then we call $\textsc{TargetSplit}(A,k,\frac{r}{2},\frac{3r}{2})$ when $A$ needs to split. Otherwise, we call $\textsc{TargetSplit}(A,k,r,\frac{3r}{2})$. 

 To analyze this block structure - we first prove that the invariant is maintained, i.e that all blocks have sizes in $\{\frac{r}{2},r,\frac{3r}{2}\}$. When the second batch begins we have a single block with $r$ keys, and we'll call $\textsc{TargetSplit}(A,k,\frac{r}{2},\frac{3r}{2})$ when $2r-B-1$ keys are remaining in the batch, and we can thus apply  \Cref{lem:TargetSplit} to argue the result is two blocks of size $\frac{r}{2}$ and $\frac{3r}{2}$ at the end of the batch. If the next batch hits the block of size $\frac{r}{2}$ it becomes a block of size $\frac{3r}{2}$. Otherwise, it hits the block of size $\frac{3r}{2}$, and we call $\textsc{TargetSplit}(A,k,r,\frac{3r}{2})$, which results in creating two blocks of size $r$ and $\frac{3r}{2}$.
Overall, the invariant is that all block sizes are in $\{\frac{r}{2},r,\frac{3r}{2}\}$.

Now, let $x_{1}^{n},x_{2}^{n},x_{3}^{n}$ denote the number of blocks of size $\frac{r}{2},r,\frac{3r}{2}$ after $n$ insertions have occurred. Further, let $v_{n}=\left(\mathbb{E}\left(x_{1}^{n}\right),\mathbb{E}\left(x_{2}^{n}\right),\mathbb{E}\left(x_{3}^{n}\right)\right)$.
Similar to \Cref{subsec:UniformModeling},
we can obtain neat scalar recurrences for each component of $v_n$, which results in the following matrix recurrence: $v_{n+r}=\left(I+\frac{1}{n}\cdot A\right)\cdot v_{n}$ where $A=\begin{Bmatrix}-\frac{r}{2} & r & 0\\
0 & -r & \frac{3}{2}r\\
\frac{r}{2} & r & 0
\end{Bmatrix}$.

$A$ has eigenvalues $\lambda_{1}=r,\lambda_{2}=-r,\lambda_{3}=-1.5r$, and we can employ \Cref{thm:MatrixRecurrenceConvergance} and \Cref{subsec:AverageBlockSize} to conclude the average block size converges to $\frac{\left\langle \vec{u},w\right\rangle }{\left\langle \vec{u},\vec{1}\right\rangle }$, where $w=\left(\frac{r}{2},r,\frac{3r}{2}\right)$ and $\vec{u}$ is any eigenvector corresponding to $\lambda_{1}=r$.

We can choose $\vec{u}=\left(2,3,4\right)$, and get average block size: $\frac{\left\langle \left(2,3,4\right),\left(\frac{r}{2},r,\frac{3r}{2}\right)\right\rangle }{2+3+4}=\frac{10}{9}r$.
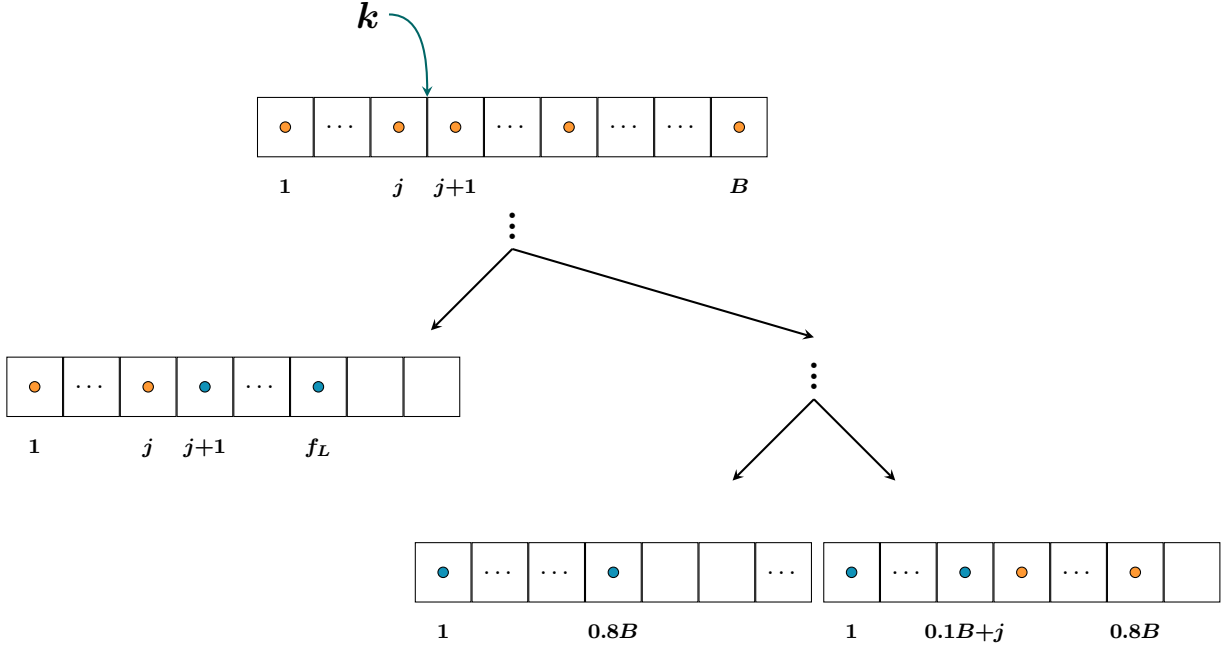
\begin{figure}[h!]
    \centering
    \resizebox{\textwidth}{!}{
    \begin{tikzpicture}[
    cell/.style={
        draw, 
        minimum width=0.75cm, 
        minimum height=0.8cm, 
        inner sep=0pt, 
        anchor=west,
        font=\small
    },
    dot/.style={
        circle, 
        draw=black, 
        inner sep=0pt, 
        minimum size=4pt
    },
    odot/.style={dot, fill=orange!80!white},
    bdot/.style={dot, fill=cyan!70!black},
    label text/.style={
        font=\bfseries\boldmath\scriptsize, 
        anchor=north,
        yshift=-0.1cm,
        text height=1.5ex,
        text depth=0.5ex
    }
]

    \coordinate (start) at (-4.125, 0); 

    \node[cell] (t1) at (start) {}; \node[odot] at (t1) {}; \node[label text] at (t1.south) {1};
    \node[cell] (t2) [right=0pt of t1] {$\dots$};
    \node[cell] (t3) [right=0pt of t2] {}; \node[odot] at (t3) {}; \node[label text] at (t3.south) {$j$};
    \node[cell] (t4) [right=0pt of t3] {}; \node[odot] at (t4) {}; \node[label text] at (t4.south) {$j{+}1$};
    \node[cell] (t5) [right=0pt of t4] {$\dots$};
    \node[cell] (t6) [right=0pt of t5] {}; \node[odot] at (t6) {}; 
    \node[cell] (t7) [right=0pt of t6] {$\dots$};
    \node[cell] (t8) [right=0pt of t7] {$\dots$};
    \node[cell] (t9) [right=0pt of t8] {}; \node[odot] at (t9) {}; \node[label text] at (t9.south) {$B$};

    \coordinate (split_point) at (t3.north east);
    \node (k_label) [above=0.8cm of split_point, xshift=-0.8cm, font=\Large\bfseries\boldmath] {$k$};
    \draw[->, thick, >=stealth, color=teal!80!black] (k_label.east) to[out=0, in=90] (split_point);

    \node (dots1) [below=0.4cm of t5, font=\Large\bfseries] {$\vdots$};
    
    \coordinate (ml_start) at (-7.5, -3.5);
    \node[cell] (l1) at (ml_start) {}; \node[odot] at (l1) {}; \node[label text] at (l1.south) {1};
    \node[cell] (l2) [right=0pt of l1] {$\dots$};
    \node[cell] (l3) [right=0pt of l2] {}; \node[odot] at (l3) {}; \node[label text] at (l3.south) {$j$};
    \node[cell] (l4) [right=0pt of l3] {}; \node[bdot] at (l4) {}; \node[label text] at (l4.south) {$j{+}1$};
    \node[cell] (l5) [right=0pt of l4] {$\dots$};
    \node[cell] (l6) [right=0pt of l5] {}; \node[bdot] at (l6) {}; \node[label text] at (l6.south) {$f_L$};
    \node[cell] (l7) [right=0pt of l6] {}; 
    \node[cell] (l8) [right=0pt of l7] {}; 

    \coordinate (dots2_pos) at (3.375, -3.25);
    \node (dots2) at (dots2_pos) [font=\Large\bfseries] {$\vdots$};

    \draw[->, thick, >=stealth] (dots1.south) -- ++(225:1.55cm);
    \draw[->, thick, >=stealth] (dots1.south) -- (dots2.north);

    \coordinate (bl2_start) at (-2.0, -6.0);
    \node[cell] (bl1) at (bl2_start) {}; \node[bdot] at (bl1) {}; \node[label text] at (bl1.south) {1};
    \node[cell] (bl2) [right=0pt of bl1] {$\dots$};
    \node[cell] (bl3) [right=0pt of bl2] {$\dots$};
    \node[cell] (bl4) [right=0pt of bl3] {}; \node[bdot] at (bl4) {}; \node[label text] at (bl4.south) {$0.8B$};
    \node[cell] (bl5) [right=0pt of bl4] {}; 
    \node[cell] (bl6) [right=0pt of bl5] {}; 
    \node[cell] (bl7) [right=0pt of bl6] {$\dots$};

    \coordinate (br2_start) at (3.5, -6.0);
    \node[cell] (br1) at (br2_start) {}; \node[bdot] at (br1) {}; \node[label text] at (br1.south) {1};
    \node[cell] (br2) [right=0pt of br1] {$\dots$};
    \node[cell] (br3) [right=0pt of br2] {}; \node[bdot] at (br3) {}; \node[label text] at (br3.south) {$0.1B{+}j$};
    \node[cell] (br4) [right=0pt of br3] {}; \node[odot] at (br4) {}; 
    \node[cell] (br5) [right=0pt of br4] {$\dots$};
    \node[cell] (br6) [right=0pt of br5] {}; \node[odot] at (br6) {}; \node[label text] at (br6.south) {$0.8B$};
    \node[cell] (br7) [right=0pt of br6] {}; 

    \draw[->, thick, >=stealth] (dots2.south) -- ++(225:1.55cm);
    \draw[->, thick, >=stealth] (dots2.south) -- ++(315:1.55cm);

\end{tikzpicture}
    }
    \caption{Example for \Cref{subsec:LargeRegime} illustrating a batch of $r=1.7B$ insertions hitting a block of size $0.7B$.}
    \label{fig:example}
\end{figure}

\subsection{Deferred even splitting for large batches}
\label{subsec:LargeRegime}
When $r>\frac{2}{3}B$, we obtain good space utilization by using deferred even splitting, as defined in \Cref{sec:intro}.
We now prove that deferred even splitting obtains good space utilization in this regime. First, note that for $\frac{2}{3}B<r\leq B$, deferred even splitting is trivial to analyze - it is an invariant that at the end of each batch all blocks have size exactly $r$. So, we now turn attention to analyzing deferred even splitting when $r>B$.
\begin{lemma}
When $r>\frac{2}{3}B$ the average fullness obtained by deferred even splitting is at least $\max\left\{\frac{2}{3}, \frac{r/B+\frac{1}{2}}{\lceil r/B+1\rceil }\right\}-\frac{1}{B}$.
\end{lemma}
\begin{proof}
First, since deferred even splitting partitions the total number of keys evenly into the minimum required number of blocks, it is a trivial invariant that at the end of each batch of insertions all blocks are at least half full. So for every batch, starting from the second, we have a total of $N\geq\frac{B}{2}+r$ keys to partition into $k=\lceil N/B\rceil$ bins. The minimum number of keys in a block at the end of the batch is therefore at least $\lfloor N/k\rfloor$, and our goal now is to lower bound this quantity. First, since $N\leq r+B$, we have $k\leq\lceil\frac{r}{B}+1\rceil$. Therefore, $N/k\geq\frac{B/2+r}{\lceil r/B+1\rceil}$ which implies $\lfloor N/k\rfloor\geq\frac{B/2+r}{\lceil r/B+1\rceil}-1$. We now proceed to show $\lfloor N/k\rfloor\geq\frac{2}{3}B-1$. First, if $k=2$, since $N>1.5B$, then $\lfloor N/k\rfloor\geq\frac{4}{3}B$. If $k\geq3$, by the minimality of $k$, we know $\left(k-1\right)\cdot B<N\le kB$. Therefore:
\[
\frac{N}{k}>\frac{\left(k-1\right)\cdot B}{k}=B\left(1-\frac{1}{k}\right)\geq B\cdot\left(1-\frac{1}{3}\right)=\frac{2}{3}B
\]
So, we can thus conclude: $\lfloor N/k\rfloor\geq\frac{2}{3}B-1$.
\end{proof} 
\section{Analysis of deferred even splitting for smaller batches} \label{Sec:EvenDeferredAnalysis}

In this section we show that, when the batch size \(r\) falls into a
certain range, the deferred even split strategy admits an exact closed form
for its limiting fullness.

Recall that under deferred even split we always insert in batches of size
\(r\), and we only split a block when a single insertion causes it to exceed
capacity \(B\). When \(r \in (B/(2i),\,B/(2i-1)]\), every block size is forced
to be a multiple of \(r\), and the largest possible multiple is \((2i-1)r\).
After rescaling by a factor of \(r\), the process on block sizes becomes
exactly Yao's even-split model with capacity \(2i-1\) and split rule
\(2i \to (i,i)\). This allows us to read off a closed form for the stationary
distribution and hence for the expected fullness:

\begin{lemma}\label{lem:DeferredEvenSplitting}
Suppose $B$ and $i$ are positive integers. For $r \in \bigl(B/(2i),\,B/(2i - 1)\bigr]$, the expected final fullness of the
deferred even split strategy is
\[
\frac{2ir}{B}\,(H_{2i} - H_{i}) .
\]
where $H_k = 1 + \frac12 + \cdots + \frac1k$ is the $k$-th harmonic number.
\end{lemma}

\begin{proof}
We first claim that throughout the process every block has size of the form
$jr$ with an integer $j$ satisfying $0 \le j < 2i$. We prove this by induction on the number of insertion steps.

Initially, there is a single empty block of size $0 = 0 \cdot r$, so the claim
holds.

Now assume that at some time every block has size $jr$ for some integer
$0 \le j < 2i$. Consider the next insertion. We pick a block of size $jr$ and
add $r$ elements to it, so its size becomes $(j+1)r$.

\begin{itemize}
  \item If $j+1 \leq 2i-1$, then
  \[
    (j+1)r \le (2i-1)r \le B,
  \]
  where the last inequality uses $r \le B/(2i-1)$. Thus no split occurs and the
  new size $(j+1)r$ is again of the required form, with $0 \le j+1 < 2i$.

  \item If $j+1 = 2i$, then the new size is $(j+1)r = 2ir$, and now
  \[
    2ir > B,
  \]
  since $r > B/(2i)$. Thus this block overflows and is split. Under the
  deferred even split strategy, we split its total size $2ir$ evenly into two
  blocks of size $ir$ and $ir$. Both are of the form $jr$ with $j=i < 2i$.
\end{itemize}

In all cases, after the insertion (and possible split) every block size is still
of the form $jr$ with $0 \le j < 2i$. This completes the induction and proves
the claim.

\medskip

Because all block sizes are multiples of $r$, we can rescale and think of each
block as holding an integer number of “units’’ instead of elements: a block of
size $jr$ elements corresponds to a block of size $j$ units. Under this
rescaling:

\begin{itemize}
  \item In each step we pick a block with probability proportional to its size
  \[
    \frac{jr}{\sum_k j_k r} = \frac{j}{\sum_k j_k},
  \]
  so in the rescaled process blocks are chosen with probability proportional to
  their current number of units $j$.

  \item When a block of size $j$ units is chosen, we add $r$ elements, i.e., the
  block gains one unit and its size becomes $j+1$.

  \item If $j+1 \le 2i-1$, the block does not overflow (because $(j+1)r \le B$),
  so no split occurs.

  \item If $j+1 = 2i$, then the block attains size $2i$ units (total size
  $2ir > B$) and is immediately split into two blocks of size $i$ units each.
\end{itemize}

Thus, in units, we obtain the following process: blocks have integer sizes
$0,1,\dots,2i-1$; at each step we pick a block of size $j$ with probability
proportional to $j$, increase its size to $j+1$, and whenever a block reaches
size $2i$ we split it into two blocks of size $i$. This is exactly the
even-split model with capacity $2i-1$ and split rule $2i \to (i,i)$ analyzed
by Yao.

For the case $r = 1$, the deferred even split strategy coincides with the even split strategy. Hence, by Lemma~\ref{lem:evensplitaveragebocksizeconverges}, the vector of block fractions converges to a stationary distribution. Let $u_j$ denote the stationary fraction of blocks of size $j$ units in this rescaled chain. By Lemma~\ref{lem:step-r-recurrence}, we have $u_{j} = \frac{j-1}{j+1}u_{j-1}$ for $j > i$. One can easily check that this is a probability distribution:
\[
u_j = 0 \quad \text{for } j < i,
\qquad
u_j = \frac{2i}{j(j+1)} \quad \text{for } j = i,i+1,\dots,2i-1.
\]

Returning to the original (unscaled) process, a block of size $j$ units has
physical size $jr$ and capacity $B$. Therefore its fullness is $jr/B$, and the
expected fullness in the steady state is
\begin{align*}
\mathbb{E}[\text{fullness}]
  &= \sum_{j=i}^{2i-1} \frac{jr}{B} \cdot u_j \\
  &= \frac{r}{B}
     \sum_{j=i}^{2i-1} j \cdot \frac{2i}{j(j+1)} \\
  &= \frac{2ir}{B}
     \sum_{j=i}^{2i-1} \frac{1}{j+1} \\
  &= \frac{2ir}{B}\,\bigl(H_{2i} - H_i\bigr),
\end{align*}
which is exactly the desired expression.
\end{proof}

Thus, for every \(i \ge 1\) and every batch size
\(r \in (B/(2i),\,B/(2i-1)]\), deferred even split has a well-defined limiting
fullness, and it is given by the simple harmonic-expression
\(\frac{2ir}{B}(H_{2i} - H_i)\). In particular, different choices of \(r\) in
this interval all induce the same underlying rescaled chain (with capacity
\(2i-1\)), and only rescale the fullness by the factor \(r/B\).

\section*{Acknowledgements}
This work was supported by NSF grants 
CNS-2504471 and
CCF-2247577.
Any opinions, findings, and conclusions or recommendations expressed in this material are those of the authors and do not necessarily reflect the views of the National Science Foundation.

Michael Bender was supported in part by the John L. Hennessy Chaired Professorship and Sandia National Laboratories.
Aaron Bernstein is funded by the NSF CAREER grant, the Sloan fellowship, and the Charles S. Baylis endowment at NYU.
Mart\'{\i}n Farach-Colton was supported in part by the Leonard J. Shustek professorship.
Hanna Koml\'os was supported in part by the Graduate Fellowships for STEM Diversity. Part of this work was supported by the Simons Institute for the Theory of Computing, and conducted when Hanna Koml\'os was visiting the Institute.

\bibliographystyle{plainurl}
\bibliography{refs}

\end{document}